\documentclass[11pt]{article}
\usepackage{palatino}
\usepackage{microtype}
\usepackage{pgfplots} 
\usepackage{tikz, tikzpeople}
\usetikzlibrary{scopes}
\usetikzlibrary{calc,intersections,decorations.pathmorphing,decorations.markings,patterns,arrows.meta,positioning}
\usetikzlibrary{angles,quotes}
\usetikzlibrary{arrows.meta}
\usetikzlibrary{arrows}
\usetikzlibrary{shapes.multipart}
\usetikzlibrary{fit}
\usetikzlibrary{shapes.geometric,positioning}
\usetikzlibrary{calc}
\usepackage{bm}
\usepackage[utf8]{inputenc}
\usepackage{mathtools}
\usepackage{xfrac}
\usepackage{graphicx}
\usepackage[linesnumbered,ruled,vlined]{algorithm2e}
\usepackage{cite}
\usepackage{amsmath,amssymb,amsfonts, amsthm,  mathrsfs}
\usepackage{algpseudocode}
\usepackage{textcomp}
\usepackage{xcolor}
\usepackage{braket}
\usepackage{enumitem}
\usepackage{cancel}
\usepackage{hyperref}
\usepackage{stfloats}

\usetikzlibrary{patterns}

\newtheorem{theorem}{Theorem}

\newtheorem{remark}{Remark}
\newtheorem{corollary}{Corollary}
\newtheorem{definition}{Definition}

\newlength{\leftstackrelawd}
\newlength{\leftstackrelbwd}
\def\leftstackrel#1#2{\settowidth{\leftstackrelawd}%
{${{}^{#1}}$}\settowidth{\leftstackrelbwd}{$#2$}%
\addtolength{\leftstackrelawd}{-\leftstackrelbwd}%
\leavevmode\ifthenelse{\lengthtest{\leftstackrelawd>0pt}}%
{\kern-.5\leftstackrelawd}{}\mathrel{\mathop{#2}\limits^{#1}}}

\newcommand{\s}{\sigma}

\newcommand{\mN}{\mathsf{N}} 
\newcommand{\mZ}{\mathsf{Z}} 
\newcommand{\mP}{\mathsf{P}}
\newcommand{\mE}{\mathsf{E}} 
\newcommand{\mD}{\mathsf{D}}

\newcommand{\sfp}{\mathsf{p}}
\newcommand{\sfq}{\mathsf{q}}

\newcommand{\C}{\mathchoice
  {\mathrm{\scriptscriptstyle C}} 
  {\mathrm{\scriptscriptstyle C}} 
  {\mathrm{\scriptscriptstyle C}} 
  {\mathrm{\scriptscriptstyle C}} 
}

\newcommand{\NS}{\mathchoice
  {\mathrm{\scriptscriptstyle NS}} 
  {\mathrm{\scriptscriptstyle NS}} 
  {\mathrm{\scriptscriptstyle NS}} 
  {\mathrm{\scriptscriptstyle NS}} 
}

\newcommand{\ca}{\mathchoice
  {\mathrm{\scriptscriptstyle ca}} 
  {\mathrm{\scriptscriptstyle ca}} 
  {\mathrm{\scriptscriptstyle ca}} 
  {\mathrm{\scriptscriptstyle ca}} 
}

\newcommand{\strc}{\mathchoice
  {\mathrm{\scriptscriptstyle sc}} 
  {\mathrm{\scriptscriptstyle sc}} 
  {\mathrm{\scriptscriptstyle sc}} 
  {\mathrm{\scriptscriptstyle sc}} 
}

\newcommand{\mixed}{\mathchoice
  {\mathrm{\scriptscriptstyle mixed}} 
  {\mathrm{\scriptscriptstyle mixed}} 
  {\mathrm{\scriptscriptstyle mixed}} 
  {\mathrm{\scriptscriptstyle mixed}} 
}

\newcommand{\nc}{\mathchoice
  {\mathrm{\scriptscriptstyle nc}} 
  {\mathrm{\scriptscriptstyle nc}} 
  {\mathrm{\scriptscriptstyle nc}} 
  {\mathrm{\scriptscriptstyle nc}} 
}

\newcommand{\opt}{\mathchoice
  {\mathrm{\scriptscriptstyle opt}} 
  {\mathrm{\scriptscriptstyle opt}} 
  {\mathrm{\scriptscriptstyle opt}} 
  {\mathrm{\scriptscriptstyle opt}} 
}

\newcommand{\R}{\mathchoice
  {\mathrm{\scriptscriptstyle  R}} 
  {\mathrm{\scriptscriptstyle  R}} 
  {\mathrm{\scriptscriptstyle  R}} 
  {\mathrm{\scriptscriptstyle  R}} 
}
 
\newcommand{\CSIR}{\mathchoice
  {\mathrm{\scriptscriptstyle  CSIR}} 
  {\mathrm{\scriptscriptstyle  CSIR}} 
  {\mathrm{\scriptscriptstyle  CSIR}} 
  {\mathrm{\scriptscriptstyle  CSIR}} 
}

\allowdisplaybreaks

\title{Non-signaling Assisted Capacity of a Classical Channel\\ with Causal CSIT}

\author{Yuhang Yao, Syed A. Jafar\\
{\small University of California Irvine, Irvine, CA 92697}\\
{\small \it Email: \{yuhangy5, syed\}@uci.edu}
}

\date{}

\begin{document}
\maketitle

\begin{abstract}
The non-signaling (NS) assisted capacity of a classical discrete memoryless channel with causal channel state information at the transmitter (CSIT) is shown to be $C^{\NS,\ca}=\max_{\mP_{\!X|S}}I(X;Y\mid S)$, where $X, Y, S$ correspond to the input, output and state of the channel. Remarkably, this is the same as the  capacity of the channel in the NS-assisted non-causal CSIT setting, $C^{\NS,\nc}=\max_{\mP_{\!X|S}}I(X;Y\mid S)$, which was previously established, and also matches the (either classical or with NS assistance) capacity of the channel where the state is available not only (either causally or non-causally) to the transmitter but also to the receiver. While the capacity remains unchanged, the optimal probability of error for fixed message  size and blocklength, in the NS-assisted causal CSIT setting can be further improved if channel state is made available to the receiver. This is in contrast to corresponding NS-assisted non-causal CSIT setting where it was previously noted that the optimal probability of error cannot be further improved by providing the state to the receiver. As a separate result we prove that non-signaling assistance, feedback, and strictly causal CSIT (i.e., transmitter knows only past channel states but not the current or future states), each of which is individually already known to not increase capacity, also cannot increase capacity when they are collectively made available to the transmitter.
\end{abstract}
\newpage
 
\section{Introduction}
The `channel with state' is a canonical setting in information theory, defined by  the parameters $(\mathcal{X},\mathcal{Y},\mathcal{S},\mN(y\mid x,s), \mP_{\!S})$  such that $\mathcal{X},\mathcal{Y},\mathcal{S}$ are the alphabet sets corresponding to the channel input, output and state, $\mP_{\!S}$ is the distribution of the state, and $\mN(y\mid x,s)$ is the probability that the channel produces output $y\in\mathcal{Y}$ given the input  $x\in\mathcal{X}$ and the channel state  $s\in\mathcal{S}$. The channel is memoryless, the state sequence is i.i.d, and depending on whether the \underline{c}hannel \underline{s}tate \underline{i}nformation at the \underline{t}ransmitter (CSIT) is available \emph{strictly} causally (only past states are known), causally (past and present states are known) or non-causally (past, present and future states are all known in advance) in time, the channel capacity values are different. The capacity with strictly causal CSIT is the same as without CSIT \cite{lapidothCommonState}, $C^{\C,\strc}(\mN,\mP_{\!S})=\max_{\mP_X}I(X;Y)$. The capacity with causal CSIT is shown by Shannon in \cite{shannon_CSIT} to be $C^{\C,\ca}(\mN,\mP_{\!S})=\max_{\mP_U\mP_{\!X\mid U,S}}I(U;Y)$, and that with non-causal CSIT is shown by Gelfand and Pinsker in \cite{gel1980coding} to be $C^{\C,\nc}(\mN,\mP_{\!S})=\max_{\mP_{\!UX\mid S}}I(U;Y)-I(U;S)$. Here the superscripts \scalebox{1.5}{$\strc,\ca, \nc, \C$} stand for `\emph{strictly causal CSIT},'  `\emph{causal CSIT},' `\emph{non-causal CSIT},' and `\emph{classical},' respectively.

The aforementioned are \emph{classical} capacity results in the sense that they allow only classical resources for encoding and decoding operations at the transmitter and receiver. These capacity results do not apply if the encoding and decoding operations are allowed to utilize quantum resources that may be shared in advance. The challenge is that while pre-shared quantum entanglement  between a transmitter and a receiver is a \emph{non-signaling} (NS) resource, i.e., it cannot by itself enable any communication between them, it still enables \emph{non-local correlations}, that are not included in the classical framework. Indeed, the capacity with shared quantum-entanglement is known to be strictly higher than the classical  capacity in some settings. For example,  in a $2$-to-$1$ classical multiple access channel (MAC) studied in \cite{leditzky2020playing}, \cite{seshadri2023separation}, when the two transmitters have entangled quantum resources, they can communicate at a sum-rate higher than the classical sum-capacity of that MAC. There are also well known cases, such as the classical point-to-point channel, where the capacity with quantum-entanglement is the same as the classical capacity \cite{Bennett_Shor_Smolin_Thapliyal_PRL}, even though the finite blocklength probability of error can be significantly improved \cite{cubitt2011zero,agarwal2024nonlocalityassistedenhancementerrorfreecommunication}. Remarkably, it is shown in \cite{matthews2012linear} that in a classical point-to-point discrete memoryless channel (without state) even if the encoder and decoder are allowed to utilize any shared non-signaling resource (which strictly includes quantum resources), the capacity (called the NS-assisted capacity) is the same as the classical capacity. 

The NS-assisted capacity formulation allows communicating parties to share any resources in advance as long as those resources by themselves are not capable of allowing the parties to communicate with each other. NS-assisted capacity is of particular interest for a number of reasons --- 1) because non-signaling correlations are much more tractable than quantum correlations, 2) because the tractability often allows capacity expressions in more computable forms, 3) because quantum-assisted capacity is sandwiched between classical and NS-assisted capacity, so the latter helps narrow down the search for settings where the biggest capacity gains from quantum resources may be found, and 4) because the NS-assisted problem formulation leads to efficient linear programming approaches for finite blocklength performance analysis. Indeed, NS-assisted capacity has been explored for various canonical networks, such as multiple access \cite{fawzi2024MAC}, broadcast \cite{fawzi2024broadcast,Yao_Jafar_NS_DoF} and interference channels \cite{Quek_Shor}. Remarkably, there are known instances of $K$-user broadcast channels for which the NS-assisted capacity is  a factor of $K$ larger than their classical sum-rate capacity \cite{Yao_Jafar_NS_DoF}. For the point-to-point channel with state, $(\mathcal{X},\mathcal{Y},\mathcal{S},\mN(y\mid x,s), \mP_{\!S})$, the NS-assisted capacity with non-causal CSIT is shown to be $C^{\NS,\nc}(\mN,\mP_{\!S})=\max_{\mP_{\!X|S}}I(X;Y\mid S)$ in \cite{Yao_Jafar_CSITTP}, which can have an unbounded multiplicative gap to classical capacity $C^{\C,\nc}(\mN,\mP_{\!S})$. Despite these significant advances, the existing knowledge remains quite limited regarding how non-local correlations enabled by non-signaling resources impact capacity results even for the elemental settings in information theory. For example, the NS-assisted capacity $C^{\NS,\ca}(\mN,\mP_{\!S})$ is unknown for a channel with state and \emph{causal CSIT}. This setting is our main focus in this work. As a separate but relatively straightforward result, we also determine the NS-assisted capacity under \emph{strictly} causal CSIT. The results are summarized in Table \ref{tab:capacity_comparison_full}.

As noted in Table \ref{tab:capacity_comparison_full}, we prove (Theorem \ref{thm:NS_capacity}) that $C^{\NS,\ca}(\mN,\mP_{\!S})=\max_{\mP_{\!X\mid S}}I(X;Y\mid S)$, which coincides with both 1) the classical capacity of a channel with state known by both the transmitter (CSIT) and the receiver (CSIR), and 2) the capacity of a channel with state with NS assistance and non-causal CSIT \cite{Yao_Jafar_CSITTP}. Evidently under NS assistance, coding with causal CSIT achieves the same capacity as that with non-causal CSIT. However, there is a clear distinction between the achievability argument in this work and that in \cite{Yao_Jafar_CSITTP}. In \cite{Yao_Jafar_CSITTP}, it is shown that the \emph{optimal probability of success} for any finite number of channel uses and any finite message size achieved with non-causal CSIT is equal to that achieved with both non-causal CSIT and CSIR. Therefore in \cite{Yao_Jafar_CSITTP} the CSIT was thought of as being `virtually signaled' to the receiver with  the measure being the probability of successful decoding (capacity thus follows). 
However, in this work, we prove (Theorem \ref{thm:Prob}) via an example of a channel with state and causal CSIT, that the NS-assisted optimal probability of success for a fixed message size and a fixed number of channel uses can be strictly improved by CSIR. Therefore, the NS-assisted coding schemes with \emph{causal CSIT} in this work are necessarily strictly weaker in general than those with non-causal CSIT in \cite{Yao_Jafar_CSITTP}. Finally, if  CSIT is strictly causal, then we show (Theorem \ref{thm:strictly_causal}) that that capacity is the same as with no CSIT and no NS-assistance.

\begin{table*}[t]
\centering
\caption{A comparison of classical and NS-assisted channel capacities with different types of CSIT. $X,Y$ and $S$ denote the input, output, and channel state, respectively.}
\label{tab:capacity_comparison_full}
\vspace{0.1cm}
\begin{tabular}{|c|c|c|}
	\hline Channel capacity & Classical (C) & Non-signaling (NS) \\ \hline No CSIT ($\emptyset$) & \begin{tabular}[c]{@{}c@{}}{\footnotesize $C^{\C,\emptyset}= \max_{\mP_{\!X}} I(X;Y)$}  \end{tabular}  & \begin{tabular}[c]{@{}c@{}}{\footnotesize $C^{\NS,\emptyset}=\max_{\mP_{\!X}} I(X;Y)$}   {\footnotesize \cite{matthews2012linear}} \end{tabular} \\ \hline Strictly causal CSIT (sc) & \begin{tabular}[c]{@{}c@{}}{\footnotesize $C^{\C,\strc}= \max_{\mP_{\!X}} I(X;Y)$} {\footnotesize \cite{lapidothCommonState}} \end{tabular}  & \begin{tabular}[c]{@{}c@{}}{\footnotesize $C^{\NS,\strc}=\max_{\mP_{\!X}} I(X;Y)$} \\ {\color{blue}\footnotesize [This work, Thm. 3]} \end{tabular}  \\\hline Causal CSIT (ca) &  \begin{tabular}[c]{@{}c@{}}{\footnotesize $C^{\C,\ca}= \max_{\mP_{\!U}, x(u,s)}I(U;Y)$} \\ {\footnotesize \cite{shannon_CSIT}, \cite[Thm. 7.2]{NIT}} \end{tabular} & \begin{tabular}[c]{@{}c@{}}{\footnotesize $C^{\NS,\ca}=\max_{\mP_{\!X\mid S}}I(X;Y\mid S)$} \\ {\color{blue}\footnotesize [\mbox{This work}, Thm. 1]} \end{tabular} \\\hline Non-causal CSIT (nc) & \begin{tabular}[c]{@{}c@{}}{\footnotesize $C^{\C,\nc}=\max_{\mP_{\!U\mid S}, x(u,s)}\big(I(U;Y)-I(U;S) \big)$} \\ {\footnotesize \cite{gel1980coding}, \cite[Thm. 7.3]{NIT}} \end{tabular} &  \begin{tabular}[c]{@{}c@{}}{\footnotesize $C^{\NS,\nc}=\max_{\mP_{\!X\mid S}}I(X;Y\mid S)$} \\ {\footnotesize \cite[Thm. 2]{Yao_Jafar_CSITTP}} \end{tabular} \\\hline
\end{tabular}
\end{table*}

\section{Preliminaries}
\subsection{Notation}
$\mathbb{N}$ denotes the set of positive integers. $\mathbb{R}_{\geq 0}$ denotes the set of non-negative reals. For $i, j \in \mathbb{N}$, $[i:j]$ denotes the set $\{i,i+1,\ldots, j\}$ if $i\leq j$ and the empty set otherwise.  $x_i^j$ is the shorthand notation for $(x_k: k\in[i:j])$. $[n]$ and $x^n$ are used to represent $[1:n]$ and $x_1^n$, respectively. 
$\mathbb{I}[x]$ denotes the indicator function, returning $1$ when the predicate $x$ is true and $0$ otherwise. $\Pr(E)$ denotes the probability of an event $E$.  
$\mathcal{X} \times \mathcal{Y}$ denotes the Cartesian product of $\mathcal{X}$ and $\mathcal{Y}$, and $\mathcal{X}^n$ denotes the $n$-fold Cartesian product of $\mathcal{X}$.

\subsection{Information-theoretic quantities}
For a finite set $\mathcal{A}$, let $\mathcal{P}(\mathcal{A})$ denote the set of probability mass functions (distributions) on $\mathcal{A}$, i.e., the set of all functions $\mP \colon \mathcal{A} \to \mathbb{R}_{\geq 0}$ such that $\sum_{a\in \mathcal{A}}\mP (a) = 1$.
For discrete sets $\mathcal{X}$ and $\mathcal{Y}$ , we use $\mathcal{P}(\mathcal{X} \mid \mathcal{Y})$ to denote the set of conditional distributions where the input variable is defined on the set $\mathcal{X}$ and the output variable is defined on the set $\mathcal{Y}$, i.e., the set of all functions $\mP \colon \mathcal{X}\times \mathcal{Y} \to \mathbb{R}_{\geq 0}$ such that for each $y\in \mathcal{Y}$, $\sum_{x\in \mathcal{X}} \mP(x\mid y) = 1$.

The notation $X\sim \mP$ indicates that the random variable $X$ has distribution $\mP$. Given that a pair of random variables $(X,Y) \sim \mP \in \mathcal{P}(\mathcal{X} \times \mathcal{Y})$, we use $\mP_{\!X}\in \mathcal{P}(\mathcal{X})$ and $\mP_{\! Y}\in \mathcal{P}(\mathcal{Y})$ to denote the marginal distributions thus defined for $X$ and $Y$, respectively, such that $\mP_{\!X}(x) = \sum_{y \in \mathcal{Y}}\mP(x,y)$ for each $x\in \mathcal{X}$, and $\mP_{\!Y}(y) = \sum_{x \in \mathcal{X}}\mP(x,y)$ for each $y\in \mathcal{Y}$.
We use $\mP_{\! X\mid Y} \in \mathcal{P}(\mathcal{X} \mid \mathcal{Y})$ to denote the conditional distribution of $X$ given $Y$, such that $\mP_{\! X\mid Y}(x\mid y) = \mP_{\! XY}(x,y)/\mP_{\! Y}(y)$ if $\mP_{\! Y}(y)\not=0$, and $\mP_{\! X\mid Y}(x\mid y)$ be any distribution in $\mathcal{P}(\mathcal{X})$ if $\mP_{\! Y}(y) =0$. The conditional distribution of $Y$ given $X$ is denoted as $\mP_{\! Y\mid X} \in \mathcal{P}(\mathcal{Y} \mid \mathcal{X}) $ and is defined similarly.
When the relevant variables are clear from the context, we sometimes omit the subscript of the distribution. For example, we may write $\mP (y\mid x)$ instead of $\mP_{\! Y\mid X}(y\mid x)$.

For a random variable $X\sim \mP$ taking values in $\mathcal{X}$ and a function $f\colon \mathcal{X} \to \mathbb{R}$, we write $\mathbb{E}_{\mP}[f(X)] = \sum_{x\in \mathcal{X}}\mP(x)  f(x)$ to denote the expectation of $f(X)$ under the distribution $\mP$.
The entropy of $X$ is then defined as $H_{\mP}(X) \triangleq -\mathbb{E}_{\mP}[\log_2(\mP(X))]$. Conditional entropy $H_{\mP}(Y\mid X)$, mutual information $I_{\mP}(X;Y)$, and conditional mutual information $I_{\mP}(X;Y\mid Z)$ are defined in the standard way. The subscript $\mP$ is often omitted when the underlying distribution is clear from the context. 
For two distributions $\sfp, \sfq \in \mathcal{P}(\mathcal{X})$, the relative entropy of $\sfp$ with respect to $\sfq$, also called the KL divergence,  is defined as $D(\sfp \Vert \sfq) \triangleq \sum_{x\in \mathcal{X}} \sfp(x)\log_2 \big(\frac{\sfp(x)}{\sfq(x)}\big) = \mathbb{E}_\sfp \big[\log_2\big(\frac{\sfp(X)}{\sfq(X)}\big)\big]$.

\subsection{Sequential non-signaling correlations} \label{sec:SNSC}
For the purpose of studying coding with causal (or strictly causal) CSIT assisted by non-signaling resources, it is useful to consider the following class of conditional distributions. These distributions are special cases of the sequential correlations studied in \cite{gallego2014nonlocality} and the Time-Ordered No-Signaling (TONS) boxes studied in \cite{TONS_boxes}.

Consider two parties, Alice and Bob, who have access to a general non-signaling resource. Such a resource may be viewed as a black box with inputs and outputs on Alice's and Bob's sides. On Alice's side, the resource sequentially admits $n$ inputs, namely $A_i\in \mathcal{A}_i$, and sequentially responds with $n$ outputs, namely $U_i\in\mathcal{U}_i$, for $i=1,2,\ldots, n$. Specifically, after Alice inputs $A_i$, the resource  produces $U_i$. On Bob's side, the resource admits an input $B\in \mathcal{B}$ and produces an output $V\in \mathcal{V}$. $\mathcal{A}_1, \cdots, \mathcal{A}_n,\mathcal{B}, \mathcal{U}_1,\cdots,\mathcal{U}_n, \mathcal{V}$ are finite sets.

Let $\mP_{\! U_1 \ldots U_n V \mid A_1 \cdots A_n, B} \in \mathcal{P}(\mathcal{U}_1\times \cdots \times \mathcal{U}_n \times \mathcal{V} \mid \mathcal{A}_1 \times \cdots \times \mathcal{A}_n \times \mathcal{B})$ denote the conditional distribution of $(U_1,\ldots, U_n,V)$ given $(A_1,\ldots, A_n,B)$. Then $\mP_{\! U_1 \ldots U_n V \mid A_1 \cdots A_n B}$ satisfies the following conditions.
\begin{enumerate}
	\item $\mP_{U_1\cdots U_n\mid A_1 \cdots A_n B}(u_1,\ldots, u_n \mid a_1,\ldots, a_n, b)$ is invariant under changes of $b\in \mathcal{B}$;
	\item $\mP_{V \mid A_1 \cdots A_n B}(v \mid a_1,\ldots, a_n, b)$ is invariant under changes of $(a_1,\ldots, a_n) \in \mathcal{A}_1\times \cdots \times \mathcal{A}_n$;
	\item For each $i\in [n-1]$, the marginal distribution $\mP_{U_1\ldots U_i \mid A_1 \cdots A_n B}(u_1,\ldots, u_i \mid a_1,\ldots, a_n, b)$ is invariant under changes of $(a_{i+1},\ldots, a_n)\in \mathcal{A}_{i+1} \times \cdots \times \mathcal{A}_{n}$
\end{enumerate}
We refer to $\mP_{\! U_1 \ldots U_n V \mid A_1 \cdots A_n, B}$ satisfying these constraints as \emph{sequential non-signaling correlations}.

\subsection{Channel with state}
We adopt the standard definition of a discrete memoryless channel with state \cite[Sec. 7.1]{NIT}. 
A channel with state is specified by (finite) alphabets $(\mathcal{X}, \mathcal{Y}, \mathcal{S})$, and a tuple $(\mN, \mP_{\! S})$ where $\mN\in \mathcal{P}(\mathcal{Y} \mid \mathcal{X}\times \mathcal{S})$ and $\mP_{\! S} \in \mathcal{P}(\mathcal{S})$.
$\mN (y\mid x,s)$ specifies the channel's conditional probability distribution, i.e., the probability of  output  $Y=y$ given the input $X=x$ and the channel state $S=s$, for $x\in \mathcal{X}, s\in \mathcal{S}, y\in \mathcal{Y}$. $\mP_{\!S}(s)$ specifies the probability of the channel state being $S=s$ for $s\in \mathcal{S}$. 
The channel is memoryless and the state is i.i.d. across channel uses. Specifically, for $n$ uses of the channel, $X^n, S^n, Y^n$ collectively denote the inputs, states, and the outputs corresponding to the $n$ channel uses, respectively. The probability distribution of $S^n$ is $\mP_{\!S^n}(s^n)= \mP_{\!S}^{\otimes n}(s^n) \triangleq \prod_{i=1}^n \mP_{\!S}(s_i)$. The channel's conditional distribution for $n$ uses of the channel is $\mN^{\otimes n}(y^n\mid x^n, s^n) \triangleq \prod_{i=1}^n \mN (y_i\mid x_i,s_i)$, for $s^n\in \mathcal{S}^n$, $x^n\in \mathcal{X}^n$, $y^n \in \mathcal{Y}^n$.

\section{NS-assisted coding with causal CSIT}
In this section, we consider the communication scenario where the transmitter and receiver are allowed to share in advance a \emph{non-signaling} resource, and the transmitter knows the channel state \emph{causally}, i.e., over each channel use the transmitter knows the past and present but not the future channel states. A message $W$ originates at the transmitter, is mapped to a sequence of $n$ symbols (a codeword) $X_1,X_2,\dots, X_n$ by an encoder utilizing the transmitter's side of the NS resource and causal CSIT, and the codeword symbols are input into the channel over $n$ channel uses. 
After $n$ channel uses, the  channel outputs $Y^n$ obtained by the receiver are mapped to $\widehat{W}$ by a decoder utilizing the receiver's side of the NS resource.

\subsection{Coding schemes} \label{sec:NSschemes}
Recall that the transmitter has $n$ sequential inputs, i.e., $(W, S_1), S_2, \ldots, S_n$, and the receiver has a single collective input, i.e., $Y^n$. We model a NS-assisted coding scheme with causal CSIT by a sequential non-signaling correlation (Section \ref{sec:SNSC}). 
Specifically, a non-signaling coding scheme with causal CSIT is specified by $(M,n)$, and a sequential non-signaling correlation 
\begin{align*}
	\mZ^{\NS,\ca} \in \mathcal{P}(\mathcal{X}^n \times [M] \mid [M]\times \mathcal{S}^n \times \mathcal{Y}^n),
\end{align*}
with $\mZ^{\NS,\ca}(x^n, \widehat{w} \mid w, s^n, y^n)$ specifying the probability of $X^n=x^n, \widehat{W}=\widehat{w}$ given $W=w, S^n=s^n, Y^n=y^n$. 
Fig. \ref{fig:scheme} illustrates a NS-assisted coding scheme with causal CSIT.

\begin{figure}[htbp]
\centering
\begin{tikzpicture}
\def\w{0.9}
\def\h{0.3}
\def\hh{0.9}
\fill[gray!20] (0,0) -- (1*\w,0) -- (1*\w,-\h) -- (2*\w,-\h) -- (2*\w,-2*\h) -- (2*\w+0.25,-2*\h)
decorate[decoration={zigzag, amplitude=2pt, segment length=6pt}] { -- (2*\w+0.75,-2*\h) }
-- (3*\w,-2*\h) -- (3*\w,-3*\h) -- (4*\w,-3*\h) -- (4*\w, 0) -- (6.5*\w, 0)
-- (6.5*\w, -3*\h-\hh) -- (2*\w+0.75, -3*\h-\hh)
decorate[decoration={zigzag, amplitude=2pt, segment length=6pt}] { -- (2*\w+0.25, -3*\h-\hh) }
-- (0*\w, -3*\h-\hh) -- cycle;
\draw[thick] (0,0) -- (1*\w,0) -- (1*\w,-\h) -- (2*\w,-\h) -- (2*\w,-2*\h) -- (2*\w+0.25,-2*\h)
decorate[decoration={zigzag, amplitude=2pt, segment length=6pt}] { -- (2*\w+0.75,-2*\h) }
-- (3*\w,-2*\h) -- (3*\w,-3*\h) -- (4*\w,-3*\h) -- (4*\w, 0) -- (6.5*\w, 0)
-- (6.5*\w, -3*\h-\hh) -- (2*\w+0.75, -3*\h-\hh)
decorate[decoration={zigzag, amplitude=2pt, segment length=6pt}] { -- (2*\w+0.25, -3*\h-\hh) }
-- (0*\w, -3*\h-\hh) -- cycle;

\draw[thick,gray] (1*\w,-\h) -- (1*\w,-3*\h-\hh);
\draw[thick,gray] (2*\w,-2*\h) -- (2*\w,-3*\h-\hh);
\draw[thick,gray] (3*\w,-3*\h) -- (3*\w,-3*\h-\hh);
\draw[thick,gray] (4*\w,-3*\h) -- (4*\w,-3*\h-\hh);

\node[] (W) at (-0.75,-0.75) {\small $W$};
\draw[-{latex}, thick] (W) -- ($(W.east)+(0.45,0)$);
\node[rectangle, draw,  thick, fill=black!10] at (-1.2,-2.5) (PS1) {$\mP_{\!S}$};
\node[rectangle, draw, thick, fill=black!10] at (-1.2,-3.25) (PS2) {$\mP_{\!S}$};
\node at (-1.2,-3.75) {\small $\vdots$};
\node[rectangle, draw, thick, fill=black!10] at (-1.2,-4.5) (PS3) {$\mP_{\!S}$};

\node[rectangle, draw,  thick, minimum width = 0.75cm, minimum height=0.6cm, fill=black!10] at (3.8,-2.4) (N1) {$\mN$};
\node[rectangle, draw, thick, minimum width = 0.75cm, minimum height=0.6cm, fill=black!10] at (3.8,-3.15) (N2) {$\mN$};
\node at (3.8,-3.68) {\small $\vdots$};
\node[rectangle,  draw, thick, minimum width = 0.75cm, minimum height=0.6cm, fill=black!10] at (3.8,-4.4) (N3) {$\mN$};

\draw [thick, -{latex}] (PS1.north) -- ($(PS1.north) + (0,2.75)$) -- ($(PS1.north) + (1.65,2.75)$) -- ($(PS1.north) + (1.65,2.25)$);

\draw [thick, -{latex}] (PS2.west) -- ($(PS2.west) + (-0.3,0)$) -- ($(PS2.west) + (-0.3,4)$) -- ($(PS2.west) + (2.9,4)$) -- ($(PS2.west) + (2.9,3)$);

\draw [thick, -{latex}] (PS3.west) -- ($(PS3.west) + (-1,0)$) -- ($(PS3.west) + (-1,5.6)$) -- ($(PS3.west) + (4.7,5.6)$) -- ($(PS3.west) + (4.7,3.65)$);

\draw [thick, -{latex}] (PS1.east) -- ($(PS1.east)+(4.28,0)$);
\draw [thick, -{latex}] (PS2.east) -- ($(PS2.east)+(4.28,0)$);
\draw [thick, -{latex}] (PS3.east) -- ($(PS3.east)+(4.28,0)$);

\draw [thick, -{latex}] (0.5*\w, -3*\h-\hh) -- (0.5*\w, -3*\h-\hh-0.5) -- (0.5*\w+2.98, -3*\h-\hh-0.5); 

\draw [thick,  -] (1.5*\w, -3*\h-\hh) -- (1.5*\w, -3*\h-\hh-0.4);
\draw[thick] (1.5*\w, -3*\h-\hh-0.4)
    arc[start angle=135, end angle=225, radius=0.3];
\draw[thick, -{latex}] (1.5*\w, -3*\h-\hh-0.82) -- (1.5*\w, -3*\h-\hh-1.25) -- (1.5*\w+2.07, -3*\h-\hh-1.25);

\draw [thick] (3.5*\w, -3*\h-\hh) -- (3.5*\w, -3*\h-\hh-0.4);
\draw[thick] (3.5*\w, -3*\h-\hh-0.4)
    arc[start angle=150, end angle=210, radius=1.15];
\draw[thick, -{latex}] (3.5*\w, -3*\h-\hh-1.55) -- (3.5*\w, -3*\h-\hh-2.5) -- (3.5*\w+0.28, -3*\h-\hh-2.5);

\draw[thick, -{latex}] (N1.east) -- ($(N1.east)+(0.25,0)$) -- ($(N1.east)+(0.25,0.6)$);

\draw[thick, -{latex}] (N2.east) -- ($(N2.east)+(0.65,0)$) -- ($(N2.east)+(0.65,1.34)$);

\draw[thick, -{latex}] (N3.east) -- ($(N3.east)+(1.35,0)$) -- ($(N3.east)+(1.35,2.59)$);

\node (Wh) at (4.75,0.75) {\small $\widehat{W}$};
\draw[thick, -{latex}] ($(Wh.south)+(0,-0.44)$)--(Wh.south);

\node at (0.23,-2.1) {\small $X_1$};
\node at (1.1,-2.8) {\small $X_2$};
\node at (2.2,-2.1) {\small $\cdots$};
\node at (2.9,-3.75) {\small $X_n$};

\node at (0.7,0.3) {\small $S_1$};
\node at (1.57,0.3) {\small $S_2$};
\node at (2.25,0.3) {\small $\cdots$};
\node at (2.95,0.3) {\small $S_n$};

\node at (-0.5,-2.35) {\small $S_1$};
\node at (-0.5,-3.1) {\small $S_2$};
\node at (-0.5,-4.35) {\small $S_n$};

\node at (4.4,-2.6) {\scriptsize $Y_1$};
\node at (5.05,-2.6) {\scriptsize $Y_2$};
\node at (5.2,-2.2) {\scriptsize $\cdots$};
\node at (5.75,-2.6) {\scriptsize $Y_n$};

\node [fill=black!15, minimum width=5.2cm] at (3,-1.35) {\small $\mZ^{\NS,\ca}\big( x^n,\widehat{w} \;\big|\; [w,s^n],y^n \big)$};

\node[] at (5.4,0.3) {\sc \tiny Receiver};
\node[] at (-0.4,0.2) {\sc \tiny Transmitter};
\node[] at (3.8,-4.9) {\sc \tiny Channel};
\node[] at (-1.2,-4.9) {\sc \tiny State};

\end{tikzpicture}
\caption{NS-assisted coding scheme with causal CSIT}
\label{fig:scheme}
\end{figure}

The non-signaling and causal-CSIT constraints for $\mZ^{\NS,\ca}$ are captured by the following $3$ conditions, namely, \textit{C1}--\textit{C3}.
\begin{enumerate}[label=\textit{C\arabic*}:]
	\item The marginal probability $\mZ^{\NS,\ca}(x^n\mid w,s^n,y^n)$ is invariant under changes of $y^n$.
	\begin{equation}
	\begin{aligned}
	&\mZ^{\NS,\ca}(x^n\mid w,s^n,y^n)=\mZ^{\NS,\ca}(x^n\mid w,s^n,{y'}^n) \\
	&\forall (x^n,w,s^n,y^n, {y'}^n)\in\mathcal{X}^n\times[M]\times\mathcal{S}^n\times\mathcal{Y}^n\times\mathcal{Y}^n
	\end{aligned}
	\end{equation}
	\item The marginal probability $\mZ^{\NS,\ca}(\widehat{w}\mid w,s^n,y^n)$ is invariant under changes of $(w, s^n)$.
	\begin{equation}
	\begin{aligned}	
	&\mZ^{\NS,\ca}(\widehat{w}\mid w,s^n,y^n)=\mZ^{\NS,\ca}(\widehat{w}\mid w',{s'}^n,y^n)\\
	&\forall (\widehat{w},w,w',s^n,{s'}^n,y^n)\in[M]\times[M]\times[M]\times\mathcal{S}^n\times\mathcal{S}^n\times\mathcal{Y}^n
	\end{aligned}
	\end{equation}
	\item For $i\in [n-1]$, the marginal probability $\mZ^{\NS, \ca}(x^i, \widehat{w} \mid w, s^n, y^n)$ is invariant under changes of $s_{i+1}^n$.
	\begin{equation}
	\begin{aligned}
	&\mZ^{\NS, \ca}(x^i, \widehat{w} \mid w, s^i,s_{i+1}^n, y^n)=\mZ^{\NS, \ca}(x^i, \widehat{w} \mid w, s^i, {s'}_{i+1}^n, y^n)\\
	&\forall (i,x^i,\widehat{w},w,s^n,{s'}_{i+1}^n,y^n)\in[n-1]\times\mathcal{X}^i\times[M]\times[M]\times\mathcal{S}^n\times\mathcal{S}^{n-i}\times\mathcal{Y}^n
	\end{aligned}
	\end{equation}
\end{enumerate}
Condition \textit{C1} says that using the box only, the transmitter should not infer any information about the input at the receiver, i.e., $Y^n$. 
Condition \textit{C2} says that using the box only  (without the channel), the receiver should not infer any information about the input at the transmitter, i.e., $(W,S^n)$.
Condition \textit{C3}  says that, using the box only, the transmitter (up to the $i^{th}$ time slot), even if collaborating with the receiver, should not infer any information about the future inputs of the channel states, i.e., $S_{i+1}^n$. 
These $3$ conditions make sure that the box obeys the non-signaling and the causal CSIT assumptions.

Note that \textit{C1} and \textit{C3} together imply that the marginal probability $\mZ^{\NS,\ca}(x^i\mid w, s^n, y^n)$ is invariant under changes of $(s_{i+1}^n, y^n)$, for each $i\in [n-1]$.  In other words, using the box only, the transmitter (up to the $i^{th}$ time slot) should not infer any information about the input at the receiver, or any information about the future inputs of the channel states, i.e., $(S_{i+1}^n, Y^n)$. Following this, we write $\mZ^{\NS,\ca}(x^i\mid w,s^i)$ for the conditional probability of the scheme producing $X^i=x^i$ given $W=w$ and $S^i=s^i$, since this does not depend on $(s_{i+1}^n, y^n)$. It then follows that $\mZ^{\NS,\ca}$ admits the following factorization
\begin{align}
	&\mZ^{\NS, \ca} (x^n,\widehat{w} \mid w,s^n,y^n)\notag \\
	&= \Bigg(\prod_{i=1}^n \mZ^{\NS, \ca}(x_i\mid x^{i-1}, w,s^i)\Bigg) \mZ^{\NS, \ca}(\widehat{w} \mid x^n,w,s^n, y^n)
\end{align}

The coding scheme operates in the following manner. $(W,S^n)$ are generated first. At the first channel use, the transmitter provides $(W,S_1)$ to the scheme, which produces $X_1$ according to $\mZ^{\NS, \ca}(x_1\mid w, s_1)$. $X_1$ is sent though the channel which produces $Y_1$ according to $\mN(y_1\mid x_1,s_1)$. At the $i=2,\ldots, n$ channel use, the transmitter provides $S_i$ to the scheme, which produces $X_i$ according to $\mZ^{\NS,\ca}(x_i\mid x^{i-1}, w,s^i)$. $X_i$ is sent through the channel which produces $Y_i$ according to $\mN(y_i\mid x_i,s_i)$. After $n$ uses of the channel, the receiver provides $Y^n$ to the scheme, which produces the decoded message $\widehat{W}$ according to $\mZ^{\NS, \ca}(\widehat{w} \mid x^n,w,s^n, y^n)$.

Let $\sfp(w,x^n,y^n,s^n,\widehat{w}) = \Pr(W=w, S^n=s^n, X^n=x^n, Y^n=y^n,  \widehat{W}=\widehat{w})$ denote the joint distribution of $(W,S^n,X^n,Y^n, \widehat{W})$. It is defined as follows.
\begin{align}
	&\sfp(w,s^n,x^n,y^n,\widehat{w}) \notag \\
	&= \sfp(w,s^n) \Bigg(\prod_{i=1}^n  \sfp(x_i\mid y^{i-1}, x^{i-1}, s^n,w) \sfp(y_i\mid y^{i-1},x^i,s^n,w) \Bigg) \sfp(\widehat{w} \mid y^n,x^n,s^n,w)\\
	&= \frac{1}{M} \mP_{\!S}^{\otimes n}(s^n) \Bigg(\prod_{i=1}^n \mZ^{\NS, \ca}(x_i\mid x^{i-1}, w,s^i) \mN(y_i\mid x_i, s_i) \Bigg) \mZ^{\NS, \ca}(\widehat{w} \mid x^n,w,s^n, y^n) \\
	&=\frac{1}{M}\Bigg(\prod_{i=1}^n \mP_{\! S}(s_i)  \mN(y_i\mid x_i,s_i) \Bigg) \mZ^{\NS, \ca} (x^n,\widehat{w} \mid w,s^n,y^n) \label{eq:joint_prob_causal}
\end{align}
The form of the joint distribution in \eqref{eq:joint_prob_causal} resembles the joint distribution for other NS-assisted coding schemes, e.g., \cite{matthews2012linear, fawzi2024MAC, fawzi2024broadcast, Yao_Jafar_CSITTP}.

We use
\begin{equation}  \label{eq:prob_succ_NS}
	\eta(\mZ^{\NS,\ca}) = \sum_{w,x^n,s^n,y^n} \sfp(w,s^n,x^n,y^n,w)
\end{equation} 
to denote the probability of success associated with $\mZ^{\NS,\ca}$.
Let $\mathcal{Z}^{\NS,\ca}(M,n)$ denote the set of NS-assisted coding schemes with message size $M$ and blocklength $n$.

\begin{remark}[Non-causal CSIT]
	The form of the joint distribution in \eqref{eq:joint_prob_causal} coincides with the joint distribution for a NS-assisted coding scheme with non-causal CSIT \cite{Yao_Jafar_CSITTP}, with the only difference here that $\mZ^{\NS,\ca}$ belongs to $\mathcal{Z}^{\NS,\ca}$ which satisfies all of \textit{C1}--\textit{C3}. The conditions for \emph{causal} CSIT are explicitly those in \textit{C3}. Removing \textit{C3} gives the definition of NS-assisted coding schemes with non-causal CSIT.   We denote the set of NS-assisted coding schemes with non-causal CSIT as $\mathcal{Z}^{\NS,\nc}$.
\end{remark}

\begin{remark}[Classical coding schemes]
	Any classical coding scheme is a special case of a non-signaling coding scheme. The class of classical coding schemes with causal CSIT, denoted as $\mathcal{Z}^{\C,\ca}$, is the subset of $\mathcal{Z}^{\NS, \ca}$ consisting of those $\mZ(x^n,\hat{w}\mid w,s^n,y^n)$ that admit a factorization of the form
	\begin{align*}
		\mZ(x^n,\hat{w}\mid w,s^n,y^n)=\sum_{\ell =1}^L \lambda_\ell \mE_{\ell}(x^n\mid w,s^n) \mD_{\ell}(\widehat{w}\mid y^n),
	\end{align*}
	where $\sum_{\ell =1}^L \lambda_\ell = 1, \lambda_\ell \geq 0,  \mE_{\ell} \in \mathcal{P}(\mathcal{X}^n \mid [M] \times \mathcal{S}^n)$ and $\mD_{\ell} \in \mathcal{P}([M]\mid \mathcal{Y}^n)$ for every $\ell \in [L]$. The class of classical coding schemes with non-causal CSIT, denoted by $\mathcal{Z}^{\C,\nc}$, is defined similarly.
\end{remark}

\subsection{Rate and capacity} \label{sec:def_rate_capacity}
A (communication) rate $R\in \mathbb{R}_{\geq 0}$ is said to be achievable by classical (resp. NS-assisted) coding schemes with causal (resp. non-causal) CSIT if (and only if) there exists a sequence of coding schemes $\{\mZ_n\}_{n\in \mathbb{N}}$ from the corresponding class such that
\begin{align}
	\lim_{n\to \infty} \eta(\mZ_n) = 1 ~~\mbox{and} ~~\lim_{n\to \infty} \frac{\log_2 M_n}{n} \geq R.
\end{align}
For a channel with state $(\mN, \mP_{\!S})$, let $C^{\C,\ca}(\mN, \mP_{\!S}), C^{\C,\nc}(\mN, \mP_{\!S}), C^{\NS,\ca}(\mN, \mP_{\!S})$, and $C^{\NS,\nc}(\mN, \mP_{\!S})$ denote the classical capacity with causal CSIT, the classical capacity with non-causal CSIT, the NS-assisted capacity with causal CSIT, and the NS-assisted capacity with non-causal CSIT, respectively.

\subsection{Optimal probability of successful decoding}
For $M,n\in \mathbb{N}$, the optimal probabilities of successful decoding (or the probabilities of success in short) for the classical/NS-assisted coding schemes with causal/non-causal CSIT are defined as follows.
\begin{equation}
\begin{aligned}
{\eta}^{\C, \ca}_{\opt,M,n} \triangleq \sup_{\mZ \in \mathcal{Z}^{\C, \ca}(M,n)} \eta(\mZ),&& {\eta}^{\C, \nc}_{\opt,M,n} \triangleq \sup_{\mZ \in \mathcal{Z}^{\C, \nc}(M,n)} \eta(\mZ),\\
{\eta}^{\NS, \ca}_{\opt,M,n} \triangleq \sup_{\mZ \in \mathcal{Z}^{\NS, \ca}(M,n)} \eta(\mZ),&&{\eta}^{\NS, \nc}_{\opt,M,n} \triangleq \sup_{\mZ \in \mathcal{Z}^{\NS, \nc}(M,n)} \eta(\mZ).
\end{aligned}
\end{equation}

\subsection{Results for NS-Assisted Coding with Causal CSIT} \label{sec:results}
Our main result, presented in the following theorem, is a characterization of the capacity of NS-assisted coding schemes with causal CSIT for any channel with state $(\mN , \mP_{\!S})$. 
\begin{theorem} \label{thm:NS_capacity}
	With causal CSIT, the NS-assisted capacity for the channel with state $(\mN , \mP_{\!S})$ is 
	\begin{align} \label{eq:NS_capacity}
		C^{\NS,\ca}(\mN , \mP_{\!S}) = \max_{\mP_{\!X\mid S}} I(X;Y\mid S),
	\end{align}
	where the maximization is over all $\mP_{\! X\mid S}\in \mathcal{P}(\mathcal{X}\mid \mathcal{S})$ such that $(S,X,Y)\sim  \mP_{\!S}(s) \mP_{\!X\mid S}(x\mid s) \mN(y\mid x,s)$.
\end{theorem}
\noindent Note that in \cite{Yao_Jafar_CSITTP} it is proved that $C^{\NS,\nc}(\mN , \mP_{\!S}) = \max_{\mP_{\!X\mid S}} I(X;Y\mid S)$. Since non-causal CSIT cannot be worse than causal CSIT,  we have $C^{\NS, \ca}(\mN , \mP_{\!S}) \leq C^{\NS, \nc}(\mN , \mP_{\!S})$ for any channel with state. Therefore, to prove Theorem \ref{thm:NS_capacity}, we only need to prove the achievability, i.e., $C^{\NS, \ca}(\mN , \mP_{\!S}) \geq \max_{\mP_{\!X\mid S}} I(X;Y\mid S)$. This proof is presented in Section \ref{proof:NS_capacity}. The proof requires a  novel construction of a set of NS-assisted coding schemes which satisfies both the non-signaling and the causality conditions. In particular, our coding schemes first transform each state sequence $(S_1,S_2,\ldots, S_n)$ to a state sequence $(\tilde{S}_1,\tilde{S}_2,\ldots, \tilde{S}_n)$ (processing in a causal order), and guarantee that $\tilde{S}^n$ will always have the same \emph{type} for all $S^n$. The schemes generate $X_i$ according to a distribution that depends only on $\tilde{S}_i$. The schemes also use an \emph{authentication} process inspired by the twirling steps used in the NS-assisted coding literature (e.g.,\cite{cubitt2011zero, matthews2012linear, fawzi2024MAC, fawzi2024broadcast}). The authentication process checks whether a given $y^n$ at the receiver satisfies certain joint typicality conditions with $X^n$, based on which it provides  either the correct message, or an incorrect message to the receiver. A toy example under an alternative (simplified) problem formulation is provided in Appendix \ref{sec:toyexample} to illustrate the \emph{authentication} aspect.

The following observations center around Theorem \ref{thm:NS_capacity},
 together with the results noted in Table \ref{tab:capacity_comparison_full}.
\begin{enumerate}[label=\textit{O\arabic*:}]
	\item Under NS assistance, having causal CSIT is sufficient to achieve the same capacity as with non-causal CSIT. However, having only strictly causal CSIT can only achieve the same capacity as with no CSIT.  Effectively, under NS assistance, the availability of \emph{current} CSIT is as helpful as non-causal CSIT in terms of the channel capacity. In contrast, classically, i.e., without NS assistance, causal CSIT is generally weaker than non-causal CSIT in terms of the channel capacity.
	\item Let $(\mN^{\otimes k}, \mP_{\!S}^{\otimes k})$ denote a block of $k$ parallel uses of the channel with state $(\mN, \mP_{\!S})$. In the classical case, the capacity for the block channel is then $C^{\C,\ca}(\mN^{\otimes k}, \mP_{\!S}^{\otimes k}) = \max_{\mP_U, x^k(u,s^k)}I(U;Y^k)$. By definition, in the limit of $k$, this capacity normalized by $k$ is, $\lim_{k\to\infty}\frac{1}{k}C^{\C,\ca}(\mN^{\otimes k},\mP_{\!S}^{\otimes k}) = C^{\C,\nc}(\mN,\mP_{\!S})$. For channels where $C^{\C,\ca}(\mN,\mP_{\!S}) < C^{\C,\nc}(\mN,\mP_{\!S})$, we have that $C^{\C,\ca}(\mN,\mP_{\!S})$ is superadditive, meaning that $C^{\C,\ca}(\mN^{\otimes k},\mP_{\!S}^{\otimes k}) > kC^{\C,\ca}(\mN,\mP_{\!S})$ for some $k$ (in fact for all sufficiently large $k$). In contrast, in the NS-assisted case, since $C^{\NS, \ca}(\mN,\mP_{\!S}) = C^{\NS, \nc}(\mN,\mP_{\!S}) = \max_{\mP_{\!X\mid S}} I(X;Y\mid S)$ for every channel with state, $C^{\NS, \ca}(\mN,\mP_{\!S})$ is always additive, meaning that $C^{\NS,\ca}(\mN^{\otimes k},\mP_{\!S}^{\otimes k}) = kC^{\NS,\ca}(\mN,\mP_{\!S})$ for all $k$.
\end{enumerate}

Following \textit{O1}, it is natural to ask whether causal CSIT can achieve the same optimal \emph{probability of success} as non-causal CSIT, for any given blocklength $n$ and message size $M$, under NS assistance. We answer  this question  in the negative, with a corollary (Corollary \ref{cor:prob_nc_ca}) that follows from our study of the optimal probability of success for a channel with CSIR (channel state information at the receiver, cf. \cite[Sec. 7.4.1]{NIT}).
We define a channel $\mN^{\CSIR}$ as a channel obtained from $\mN$ by providing the state explicitly to the receiver. See Definition \ref{def:CSIR}.

\begin{definition}[Channel with CSIR] \label{def:CSIR}
Given a channel with state $(\mN, \mP_{\!S})$, the associated channel with CSIR, denoted $(\mN^{\CSIR},\mP_{\!S})$, which lies in the framework of channels with state, is defined such that
\begin{align}
	\mN^{\CSIR}([y,s^{\R}] \mid x,s) = \mN(y\mid x,s) \times \mathbb{I}[s^{\R}=s],
\end{align}
for all $x\in \mathcal{X}, s\in \mathcal{S},   (y,s^{\R}) \in \mathcal{Y} \times \mathcal{S}$.
\end{definition}
\noindent In other words, $\mN^{\CSIR}$  includes the state $S$ in the output to receiver, so that the output of $\mN^{\CSIR}$ is $[Y,S^{\R}]$ where $S^{\R}=S$. Unlike CSIT, where we define causal CSIT and non-causal CSIT, for CSIR there is no such distinction in terms of causality, as the receiver is allowed to decode the message after collecting the channel's outputs for all channel uses.

\begin{remark}[CSIR] \label{rem:CSIR}
For a channel with state $(\mN, \mP_{\!S})$, we note that the expression in \eqref{eq:NS_capacity} coincides with the classical capacity when both CSIT and CSIR are available \cite[Sec. 7.4.1]{NIT}, regardless of whether the CSIT is causal or non-causal. Specifically, the four quantities, $C^{\NS,\ca}(\mN, \mP_{\!S})$, $C^{\NS,\nc}(\mN, \mP_{\!S})$, $C^{\C,\ca}(\mN^{\CSIR}, \mP_{\!S})$ and $C^{\C,\nc}(\mN^{\CSIR}, \mP_{\!S})$, are all equal to the value given in \eqref{eq:NS_capacity}. 
\end{remark}

Our next theorem shifts the focus from the capacity to the optimal probability of success for any given blocklength $n$ and message size $M$. In particular, we explore whether CSIR can further increase the probability of success of NS-assisted coding schemes, provided only causal CSIT is available. 

Given a channel with state $(\mN, \mP_{\!S})$ and its associated channel with CSIR, $(\mN^{\CSIR}, \mP_{\!S})$, a known result \cite[Thm. 1]{Yao_Jafar_CSITTP} states that,
\begin{align} \label{eq:VSCSIT}
	{\eta}^{\NS,\nc}_{\opt,M,n}(\mN, \mP_{\!S}) = {\eta}^{\NS,\nc}_{\opt,M,n}(\mN^{\CSIR}, \mP_{\!S}), \forall M,n,
\end{align}
i.e., under NS assistance, the optimal probability of success with \emph{non-causal} CSIT is the same as that with \emph{non-causal} CSIT and CSIR. This is referred to as `virtual signaling of CSIT'  in \cite{Yao_Jafar_CSITTP}. Does the same relationship also hold when only \emph{causal} CSIT is available? Our next result shows that this is in general not true, by identifying a setting where ${\eta}^{\NS,\ca}_{\opt,M,n}(\mN, \mP_{\!S}) < {\eta}^{\NS,\ca}_{\opt,M,n}(\mN^{\CSIR}, \mP_{\!S})$.  Towards this end, let us now introduce the `$Z_0/Z_1$' channel.

\begin{figure}[h]
\begin{center}
\begin{tikzpicture}
\begin{scope}[local bounding box=myBox1]
\node (X) at (0,0) {\small $x$};
\node (N0) at (2,0){\small $\mN(y\mid x,s=0)$};
\node (Y) at (4,0) {\small $y$};
\node (X0) at (0,-0.5){\small $0$};
\node (X1) at (0,-1.5){\small $1$};
\node (Y0) at (4,-0.5){\small $0$};
\node (Y1) at (4,-1.5){\small $1$};
\draw (X0)--(Y0) node [pos=0.25, below]{\small$1$};
\draw (X1)--(Y0) node [pos=0.75, below]{\small$0.5$};
\draw (X1)--(Y1) node [pos=0.75, below]{\small$0.5$};
\end{scope}
\draw[help lines] ($(myBox1.south west)+(-0.2,-0.2)$) rectangle ($(myBox1.north east)+(0.2,0.2)$);
\node at (2,-2.5) {\small $Z_0$};

\begin{scope}[shift={(5,0)}]
\begin{scope}[local bounding box=myBox2]
\node (X) at (0,0) {\small $x$};
\node (N0) at (2,0){\small $\mN(y\mid x,s=1)$};
\node (Y) at (4,0) {\small $y$};
\node (X0) at (0,-0.5){\small $0$};
\node (X1) at (0,-1.5){\small $1$};
\node (Y0) at (4,-0.5){\small $0$};
\node (Y1) at (4,-1.5){\small $1$};
\draw (X0)--(Y0) node [pos=0.75, below]{\small$0.5$};
\draw (X0)--(Y1) node [pos=0.25, below]{\small$0.5$};
\draw (X1)--(Y1) node [pos=0.25, below]{\small$1$};
\end{scope}
\draw[help lines] ($(myBox2.south west)+(-0.2,-0.2)$) rectangle ($(myBox2.north east)+(0.2,0.2)$);
\node at (2,-2.5) {\small $Z_1$};
\end{scope}
\end{tikzpicture}
\caption{The $Z_0/Z_1$ channel acts as the channel $Z_0$ when the state is $s=0$, and as the channel $Z_1$ when the state is $s=1$. The state is assumed equally likely to be $0$ or $1$.}
\end{center}
\end{figure}

\begin{definition}[$Z_0/Z_1$ channel]\label{def:Z0Z1}
	The $Z_0/Z_1$ channel is defined by $\mathcal{Y}=\mathcal{X}=\mathcal{S}=\{0,1\}$,
	\begin{equation}
	\begin{aligned}
		&\mN(0\mid 0,0) = 1, ~ \mN(1\mid 0,0) = 0, ~ \mN(0\mid 1,0) = \mN(1\mid 1,0) = 1/2, \\
		&\mN(1\mid 1,1) = 1, ~ \mN(0\mid 1,1) = 0, ~ \mN(0\mid 0,1) = \mN(1\mid 0,1) = 1/2,
	\end{aligned}
	\end{equation}
	and
	\begin{align}
		\mP_{\!S}(s)= 1/2, ~ \forall s\in \mathcal{S}.
	\end{align}
\end{definition}

\begin{theorem} \label{thm:Prob}
Let $(\mN, \mP_{\!S})$ be the $Z_0/Z_1$ channel. Then,
\begin{align}
	& {\eta}^{\NS,\ca}_{\opt,M=2,n=2}(\mN^{\CSIR}, \mP_{\!S}) \geq {\eta}^{\C,\ca}_{\opt,M=2,n=2}(\mN^{\CSIR}, \mP_{\!S}) \geq 7/8, \label{eq:Prob_1}\\
	& {\eta}^{\NS,\ca}_{\opt,M=2,n=2}(\mN, \mP_{\!S}) \leq 13/16 \label{eq:Prob_2}
\end{align}
\end{theorem}
\noindent The proof appears in Appendix \ref{proof:Prob}. Theorem \ref{thm:Prob} shows that the notion of `virtual signaling of CSIT' via NS assistance no longer holds for causal CSIT, because for the $Z_0/Z_1$ channel, CSIR can further improve the optimal probability of success beyond what is achievable with NS assistance. 

\begin{corollary} \label{cor:prob_nc_ca}
	Let $(\mN, \mP_{\!S})$ be the $Z_0/Z_1$ channel. Then,
	\begin{align}
		{\eta}^{\NS,\nc}_{\opt,M=2,n=2}(\mN, \mP_{\!S}) >  {\eta}^{\NS,\ca}_{\opt,M=2,n=2}(\mN, \mP_{\!S}).
	\end{align}
\end{corollary}
\begin{proof}
	${\eta}^{\NS,\nc}_{\opt,M=2,n=2}(\mN, \mP_{\!S}) \stackrel{\eqref{eq:VSCSIT}}{\geq} {\eta}^{\NS,\nc}_{\opt,M=2,n=2}(\mN^{\CSIR}, \mP_{\!S}) \geq {\eta}^{\NS,\ca}_{\opt,M=2,n=2}(\mN^{\CSIR}, \mP_{\!S}) \stackrel{\eqref{eq:Prob_1}}{\geq} 7/8 > 13/16 \stackrel{\eqref{eq:Prob_2}}{\geq} {\eta}^{\NS,\ca}_{\opt,M=2,n=2}(\mN, \mP_{\!S})$.
\end{proof}

\section{Extension: Strictly Causal and Non-Causal CSIT}
Having established in Theorem \ref{thm:NS_capacity} that under NS assistance, having causal CSIT suffices to achieve the same capacity as having non-causal CSIT, in this section we study the NS-assisted capacity in settings where the CSIT can be strictly causal, non-causal or a mix of both.

\subsection{Coding schemes} \label{sec:NSschemes_mixed}
To have a model that encapsulates all results we have in this section, let us consider a scenario where two kinds of CSITs, \emph{strictly causal} CSIT and \emph{non-causal} CSIT, are present. We model the channel with state by having two independent states, namely $(S,T)$, where $S$ is known non-causally to the transmitter and $T$ is known strictly causally to the transmitter. Specifically, let $\mP_{\!S}\in \mathcal{P}(\mathcal{S})$ and $\mP_{\!T}\in \mathcal{P}(\mathcal{T})$ denote the distributions for $S$ and $T$, respectively. Let $\mN\in \mathcal{P}(\mathcal{Y}\mid \mathcal{X}\times \mathcal{S} \times \mathcal{T})$. A message $W$ originates at the transmitter, and is encoded into $X_1,X_2,\dots, X_n$ that are input into the channel $\mN$ over $n$ channel uses. For the $i^{th}$ uses of the channel, the codeword symbol $X_i$ can depend on $(W, S^n, T^{i-1})$.  The encoder and decoder are allowed to share in advance any  non-signaling resource. Fig. \ref{fig:scheme_mixed} illustrates this scenario. 

\begin{figure}[htbp]
\centering
\begin{tikzpicture}
\def\w{0.9}
\def\h{0.3}
\def\hh{0.9}
\fill[gray!20] (0,0) -- (1*\w,0) -- (1*\w,-\h) -- (2*\w,-\h) -- (2*\w,-2*\h) -- (2*\w+0.25,-2*\h)
decorate[decoration={zigzag, amplitude=2pt, segment length=6pt}] { -- (2*\w+0.75,-2*\h) }
-- (3*\w,-2*\h) -- (3*\w,-3*\h) -- (4*\w,-3*\h) -- (4*\w, 0) -- (6.5*\w, 0)
-- (6.5*\w, -3*\h-\hh) -- (2*\w+0.75, -3*\h-\hh)
decorate[decoration={zigzag, amplitude=2pt, segment length=6pt}] { -- (2*\w+0.25, -3*\h-\hh) }
-- (0*\w, -3*\h-\hh) -- cycle;
\draw[thick] (0,0) -- (1*\w,0) -- (1*\w,-\h) -- (2*\w,-\h) -- (2*\w,-2*\h) -- (2*\w+0.25,-2*\h)
decorate[decoration={zigzag, amplitude=2pt, segment length=6pt}] { -- (2*\w+0.75,-2*\h) }
-- (3*\w,-2*\h) -- (3*\w,-3*\h) -- (4*\w,-3*\h) -- (4*\w, 0) -- (6.5*\w, 0)
-- (6.5*\w, -3*\h-\hh) -- (2*\w+0.75, -3*\h-\hh)
decorate[decoration={zigzag, amplitude=2pt, segment length=6pt}] { -- (2*\w+0.25, -3*\h-\hh) }
-- (0*\w, -3*\h-\hh) -- cycle;

\draw[thick,gray] (1*\w,-\h) -- (1*\w,-3*\h-\hh);
\draw[thick,gray] (2*\w,-2*\h) -- (2*\w,-3*\h-\hh);
\draw[thick,gray] (3*\w,-3*\h) -- (3*\w,-3*\h-\hh);
\draw[thick,gray] (4*\w,-3*\h) -- (4*\w,-3*\h-\hh);

\node[] (W) at (-0.75,-0.75) {\small $W$};
\draw[-{latex}, thick] (W) -- ($(W.east)+(0.45,0)$);

\node[] at (-0.75,-2.5) (PS1) {\small $(S_1,T_1)$};
\node[] at (-0.75,-3.25) (PS2) {\small $(S_2,T_2)$};
\node at (-0.75,-3.75) {\small $\vdots$};
\node[] at (-0.75,-4.5) (PS3) {\small $(S_n,T_n)$};

\node[rectangle, draw,  thick, minimum width = 0.75cm, minimum height=0.6cm, fill=black!10] at (3.8,-2.4) (N1) {$\mN$};
\node[rectangle, draw, thick, minimum width = 0.75cm, minimum height=0.6cm, fill=black!10] at (3.8,-3.15) (N2) {$\mN$};
\node at (3.8,-3.68) {\small $\vdots$};
\node[rectangle,  draw, thick, minimum width = 0.75cm, minimum height=0.6cm, fill=black!10] at (3.8,-4.4) (N3) {$\mN$};

\draw [thick, -{latex}] (PS1.east) -- ($(PS1.east)+(3.45,0)$);
\draw [thick, -{latex}] (PS2.east) -- ($(PS2.east)+(3.45,0)$);
\draw [thick, -{latex}] (PS3.east) -- ($(PS3.east)+(3.45,0)$);

\draw [thick, -{latex}] (0.5*\w, -3*\h-\hh) -- (0.5*\w, -3*\h-\hh-0.5) -- (0.5*\w+2.98, -3*\h-\hh-0.5); 

\draw [thick,  -] (1.5*\w, -3*\h-\hh) -- (1.5*\w, -3*\h-\hh-0.4);
\draw[thick] (1.5*\w, -3*\h-\hh-0.4)
    arc[start angle=135, end angle=225, radius=0.3];
\draw[thick, -{latex}] (1.5*\w, -3*\h-\hh-0.82) -- (1.5*\w, -3*\h-\hh-1.25) -- (1.5*\w+2.07, -3*\h-\hh-1.25);

\draw [thick] (3.5*\w, -3*\h-\hh) -- (3.5*\w, -3*\h-\hh-0.4);
\draw[thick] (3.5*\w, -3*\h-\hh-0.4)
    arc[start angle=150, end angle=210, radius=1.15];
\draw[thick, -{latex}] (3.5*\w, -3*\h-\hh-1.55) -- (3.5*\w, -3*\h-\hh-2.5) -- (3.5*\w+0.28, -3*\h-\hh-2.5);

\draw[thick, -{latex}] (N1.east) -- ($(N1.east)+(0.25,0)$) -- ($(N1.east)+(0.25,0.6)$);

\draw[thick, -{latex}] (N2.east) -- ($(N2.east)+(0.65,0)$) -- ($(N2.east)+(0.65,1.34)$);

\draw[thick, -{latex}] (N3.east) -- ($(N3.east)+(1.35,0)$) -- ($(N3.east)+(1.35,2.59)$);

\node (Wh) at (4.75,0.75) {\small $\widehat{W}$};
\draw[thick, -{latex}] ($(Wh.south)+(0,-0.44)$)--(Wh.south);

\node at (0.23,-2.1) {\small $X_1$};
\node at (1.1,-2.8) {\small $X_2$};
\node at (2.2,-2.1) {\small $\cdots$};
\node at (2.9,-3.75) {\small $X_n$};

\node (Sn) at (0.5,0.6) {\small $S^n$};
\draw[-{latex}, thick] (Sn) -- ($(Sn.south)+(0,-0.35)$);
\node (T1) at (1.37,0.6) {\small $T_1$};
\draw[-{latex}, thick] (T1) -- ($(T1.south)+(0,-0.65)$);
\node at (2.25,0.6) {\small $\cdots$};
\node (T2) at (3.15,0.6) {\small $T_{n-1}$};
\draw[-{latex}, thick] (T2) -- ($(T2.south)+(0,-1.2)$);

\node at (4.4,-2.6) {\scriptsize $Y_1$};
\node at (5.05,-2.6) {\scriptsize $Y_2$};
\node at (5.2,-2.2) {\scriptsize $\cdots$};
\node at (5.75,-2.6) {\scriptsize $Y_n$};

\node [fill=black!15, minimum width=5.2cm] at (3,-1.35) {\small $\mZ\big(x^n,\widehat{w} \;\big|\; w,s^n,t^{n-1},y^n \big)$};

\end{tikzpicture}
\caption{NS-assisted coding scheme for a channel with two states, $(S,T)$, where $S$ is known non-causally, and $T$ is known strictly causally to the transmitter.}
\label{fig:scheme_mixed}
\end{figure}

Such a NS-assisted coding scheme is specified by $(M,n)$ and a sequential non-signaling correlation 
\begin{align*}
	\mZ \in \mathcal{P}(\mathcal{X}^n \times [M] \mid [M] \times \mathcal{S}^{n} \times \mathcal{T}^{n-1}  \times \mathcal{Y}^n)
\end{align*}
with $\mZ (x^n,\widehat{w} \mid w, s^{n}, t^{n-1},  y^n)$ specifying the probability of $X^n=x^n, \widehat{W} = \widehat{w}$ given $W=w, S^{n}=s^{n}, T^{n-1} = t^{n-1}, Y^n = y^n$. The additional conditions for $\mZ$ are captured as follows.
\begin{enumerate}[label=\textit{C\arabic*}:, start=4]
	\item The marginal probability $\mZ(x^n \mid w, s^{n}, t^{n-1}, y^n)$ is invariant under changes of $y^n$. We write
	\begin{equation} \label{eq:cond_strc_1}
	\begin{aligned}
	&\mZ(x^n\mid w,s^n,t^{n-1},y^n) \triangleq \mZ(x^n\mid w,s^n,t^{n-1}),\\
	&\forall (x^n,w,s^n,t^{n-1},y^n) \in \mathcal{X}^n \times [M] \times \mathcal{S}^n \times \mathcal{T}^{n-1} \times \mathcal{Y}^n
	\end{aligned}
	\end{equation}
	\item The marginal probability $\mZ(\widehat{w} \mid w, s^{n}, t^{n-1}, y^n)$ is invariant under changes of $(w, s^{n}, t^{n-1})$. We write
	\begin{equation} \label{eq:cond_strc_2}
	\begin{aligned}
		&\mZ(\widehat{w} \mid w,s^n,t^{n-1},y^n) \triangleq \mZ(\widehat{w} \mid y^n), \\
		&\forall (\widehat{w},w,s^n,t^{n-1},y^n) \in  [M] \times [M] \times \mathcal{S}^n \times \mathcal{T}^{n-1} \times \mathcal{Y}^n
	\end{aligned}
	\end{equation}
	\item For each $i\in [n-1]$, the marginal probability $\mZ(x^i,\widehat{w} \mid w, s^{n}, t^{n-1}, y^n)$ is invariant under changes of $t_{i}^{n-1}$. We write
	\begin{equation}\label{eq:cond_strc_3}
	\begin{aligned} 
		&\mZ(x^i,\widehat{w} \mid w, s^{n}, t^{n-1}, y^n) \triangleq \mZ(x^i,\widehat{w} \mid w, s^{n}, t^{i-1}, y^n) \\
		&\forall (i,x^i,\widehat{w},w,s^n,t^{n-1},y^n) \in [n-1]\times \mathcal{X}^{i} \times [M]\times [M] \times \mathcal{S}^n \times \mathcal{T}^{i-1} \times \mathcal{Y}^n
	\end{aligned}
	\end{equation}
\end{enumerate}

Let us also write $\mZ(x^i\mid w,s^n, t^{n-1},y^n) \triangleq \mZ(x^i\mid w,s^n, t^{i-1})$ for the conditional probability of the scheme producing $X^i=x^i$ given $W=w, S^n=s^n$ and $T^{i-1}=t^{i-1}$, since this does not depend on $(t_{i}^{n-1},y^n)$. It then follows that $\mZ$ admits the factorization as
\begin{align}
	&\mZ(x^n,\widehat{w}\mid w,s^n,t^{n-1},y^n) \notag \\
	&= \Bigg( \prod_{i=1}^n \mZ(x_i\mid x^{i-1}, w,s^n, t^{i-1}) \Bigg) \mZ(\widehat{w}\mid x^n, w,s^n, t^{n-1}, y^n)
\end{align}

The coding scheme operates in the following way. $(W,S^n,T^n)$ are generated first. At the first channel use, the transmitter provides $(W,S^n)$ to the scheme, which produces $X_1$ to be sent through the channel. At the $i=2,\ldots, n$ channel use, the transmitter provides $T_{i-1}$ to the scheme, which produces $X_i$ to be sent through the channel. After $n$ uses of the channel, the receiver provides $Y^n$ to the scheme, which produces the decoded message $\widehat{W}$. 

Let $\sfp(w,x^n,y^n,s^n,t^n,\widehat{w}) = \Pr(W=w,S^n=s^n,T^n=t^n,X^n=x^n,Y^n=y^n,\widehat{W}=\widehat{w})$ denote the joint distribution of $(W,S^n,T^n,X^n,Y^n, \widehat{W})$ when the channel is present. Similar to \eqref{eq:joint_prob_causal}, this is defined as,
\begin{align}
	&\sfp_{W S^n T^n X^n Y^n \widehat{W}}(w, s^n, t^n, x^n, y^n,  \widehat{w}) \notag \\
	&= \sfp(w,s^n,t^n) \Bigg(\prod_{i=1}^n  \sfp(x_i\mid y^{i-1}, x^{i-1}, t^{n},s^n,w) \sfp(y_i\mid y^{i-1},x^i,t^n,s^n,w) \Bigg) \sfp(\widehat{w} \mid y^n,x^n,t^n,s^n,w)\\
	&= \frac{1}{M}\mP_{\!S}^{\otimes n}(s^n) \mP_{\!T}^{\otimes n}(t^n) \Bigg(\prod_{i=1}^n  \mZ(x_i \mid x^{i-1}, w,s^n, t^{i-1}) \mN(y_i\mid x_i, s_i,t_i) \Bigg) \mZ(\widehat{w} \mid x^n,w,s^n,t^{n-1},y^n) \\
	&= \frac{1}{M} \Bigg(\prod_{i=1}^n \mP_{\!S}(s_i)\mP_{\!T}(t_i) \mN(y_i\mid x_i, s_i) \Bigg) \mZ^{\NS, \ca} (x^n,\widehat{w} \mid w,s^n,t^{n-1},y^n) \label{eq:joint_prob_strc}
\end{align}
Let $\mathcal{Z}^{\NS,\mixed}(M,n)$ denote this class of NS-assisted coding schemes with message size $M$ and blocklength $n$. Achievable rates and the NS-assisted capacity for this setting are defined accordingly with respect to $\mathcal{Z}^{\NS,\mixed}(M,n)$, similarly to those in Section \ref{sec:def_rate_capacity}. 

\subsection{Result}
The main result in this section is formalized in Theorem \ref{thm:strictly_causal}.
\begin{theorem} \label{thm:strictly_causal}
	Let $C^{\NS,\mixed}(\mN, \mP_{\!S}, \mP_{\!T})$ denote the NS-assisted capacity for the channel with two states $(\mN, \mP_{\!S}, \mP_{\!T})$, when $S$ is known non-causally, and $T$ is known strictly causally to the transmitter. We have,
	\begin{align}
		C^{\NS,\mixed}(\mN, \mP_{\!S}, \mP_{\!T}) = \max_{\mP_{\! X\mid S} \in \mathcal{P}(\mathcal{X}\mid \mathcal{S})} I(X;Y\mid S)
	\end{align}
	where the maximization is over all $\mP_{\!X\mid S} \in \mathcal{P}(\mathcal{X}\mid \mathcal{S})$  such that $(S,T,X,Y) \sim \mP_{\!S}(s) \mP_{\!T}(t) \mP_{\!X\mid S}(x\mid s) \mN(y\mid x,s,t)$.
\end{theorem}
\noindent We present the proof in Appendix \ref{proof:strictly_causal}. Let us provide the following observations.
\begin{enumerate}[label=\textit{O\arabic*:}, start=3]
	\item Note that  $(S,X,Y) \sim \mP_{\!S}(s) \mP_{\!X\mid S}(x\mid s) \Big( \sum_{t}\mP_{\!T}(t) \mN(y\mid x,s,t) \Big) \triangleq \mP_{\!S}(s) \mP_{\!X\mid S}(x\mid s) \mN'(y\mid x,s)$, and that $I(X;Y\mid S)$  only depends on the marginal distribution of $(S,X,Y)$. Therefore, the capacity coincides with the NS-assisted capacity when $S$ is known either causally or non-causally to the transmitter, whereas $T$ is not known to the transmitter. It follows that having strictly causal CSIT cannot improve the NS-assisted capacity in any case.
	\item It follows from Theorem \ref{thm:NS_capacity}, that in order to achieve this capacity  it suffices to have only causal knowledge of $S$ at the transmitter, i.e., $S_i$ revealed to the transmitter at the $i^{th}$ channel use for $i=1,2,\ldots, n$.
	\item Since strictly causal CSIT cannot improve the capacity, it follows that feedback cannot improve the capacity either. Specifically, suppose the noise in the channel is explicitly modeled as $Z$ such that $Y = f(X,S,T,Z)$ via some function $f$. Then $Z$ is independent of $(X,S,T)$ and can be thought of as an additional channel state. For the $i^{th}$ channel use, suppose a genie provides the transmitter with $Z_{i-1}$. Then it cannot be worse than the transmitter having the feedback $Y_{i-1}$, since the transmitter can compute $Y_{i-1}=f(X_{i-1}, S_{i-1}, T_{i-1}, Z_{i-1})$. However, $(Z_{i-1}, T_{i-1})$ together can be modeled as strictly causal CSIT, which does not improve the NS-assisted capacity.
	\item Recall from prior results that non-signaling assistance, strictly causal CSIT, and feedback, when available individually, cannot increase the channel capacity. Thus, Theorem \ref{thm:strictly_causal} strengthens these results, showing that even in the presence of all three, the channel capacity cannot increase.
\end{enumerate}

\section{Proof of Theorem \ref{thm:NS_capacity}} \label{proof:NS_capacity}
In this section, we write $g(n)=o(f(n))$ if $\lim_{n\to \infty}\tfrac{g(n)}{f(n)} = 0$.
The direction $C^{\NS,\ca} \leq \max_{\mP_{\!X\mid S}} I(X;Y\mid S)$ follows from the converse for $C^{\NS,\nc}$ established in \cite{Yao_Jafar_CSITTP} because $C^{\NS,\ca} \leq C^{\NS,\nc}$. It remains to show the achievability, i.e., $C^{\NS,\ca} \geq \max_{\mP_{\!X\mid S}} I(X;Y\mid S)$. 

Fix any $\mP_{\!X\mid S}$ (which defines $\mP_{\!XYS} = \mP_{\!S}\mP_{\!X\mid S}\mN $). We will show that any rate $R<I(X;Y\mid S)$ is achievable by NS assisted coding schemes with causal CSIT (satisfying \textit{C1}--\textit{C3} in Section \ref{sec:NSschemes}). To avoid degenerate cases, let us assume $\mP_{\!S}(\sigma) >0,~ \forall \sigma \in \mathcal{S}$.

\subsection{Preliminary steps}
\noindent {\bf [Types and typicality]}
To facilitate the proof, let us first invoke the following definition of type. Let $\mathcal{A}$ be a finite alphabet (set), and let $a^n \in \mathcal{A}^n$ be an $n$-length sequence with symbols taking values in $\mathcal{A}$. Define the  (unnormalized)  type of $a^n$ as
\begin{align}
	{\bf N}_{a^n}(\alpha) \triangleq \sum_{i=1}^n \mathbb{I}[a_i=\alpha], ~~ \forall \alpha \in \mathcal{A}.
\end{align}
Note that normalizing ${\bf N}_{a^n}$ by the length $n$ gives the empirical distribution of the sequence $a^n$.

We denote by $\mathcal{T}_{\epsilon}^{(n)}(\mP_{\!A})$ (or simply $\mathcal{T}_{\epsilon}^{(n)}$ when the distribution is clear from the context) the strong typical set of length $n$ sequences \cite{NIT}, with tolerance $\epsilon$, generated according to the distribution $\mP_{\!A}$ for a random variable $A$. Specifically,
\begin{align} \label{eq:def_strong_typical_set}
	&\mathcal{T}_{\epsilon}^{(n)}(\mP_{\!A}) \triangleq\big\{a^n\in \mathcal{A}^n\colon |\tfrac{1}{n}{\bf N}_{a^n}(\alpha) - \mP_{\!A}(\alpha)| \leq \epsilon \mP_{\!A}(\alpha), \forall \alpha \in \mathcal{A} \big\}.
\end{align}
It is well known that if $A_1,A_2,\ldots, A_n$ are drawn i.i.d. according to $\mP_{\!A}$, then $\Pr\big(A^n \in \mathcal{T}_{\epsilon}^{(n)}(\mP_{\!A}) \big) \to 1$ as $n\to \infty$ for any $\epsilon>0$.

\noindent {\bf [A useful algorithm]}
Next, let us define a useful algorithm, which aims to map a sequence $a^n$ to another sequence $\tilde{a}^n$, such that type of the output ${\bf N}_{\tilde{a}^n}$ does not vary for different input $a^n$. Specifically, the inputs to the algorithm are $(n, \mathcal{A}, \mP_{\!A}, a^n, \epsilon)$, where $n\in \mathbb{N}$, $\mathcal{A}$ is a finite alphabet that does not include a special symbol $\phi$, $\mP_{\!A}$ is a distribution on $\mathcal{A}$, $a^n\in \mathcal{A}^n$, and $\epsilon\in (0,1)$. The algorithm is presented as Algorithm \ref{alg:type_mapping}. 

\begin{algorithm}[h]
\caption{Mapping $a^n\in \mathcal{A}^n$ to $\tilde{a}^n \in (\mathcal{A}\cup \{\phi\})^n$}\label{alg:type_mapping}
\KwIn{$n, \mathcal{A}, \mP_{\!A}, a^n, \epsilon$}
\KwOut{$\tilde{a}^n$, ${\sf flag}$}
\For{$\alpha \in \mathcal{A}$}
	{$c_{\alpha} \gets 0$ \Comment{Initialize a counter for each $\alpha \in \mathcal{A}$}\\
	$t_{\alpha} \gets  \lfloor n(1-\epsilon)\mP_{\!A}(\alpha) \rfloor$\\ \Comment{Set a budget for each $\alpha \in \mathcal{A}$}}
$\bar{c} \gets 0$ \Comment{Initialize an extra counter} \\
$\bar{t} \gets n - \sum_{\alpha}t_{\alpha}$ 
\Comment{Note that $\sum_{\alpha}t_{\alpha} + \bar{t} = n$}\\
${\sf flag} \gets 1$\\
\For{$i \gets 1,2,\ldots, n$}{
	\If{${\sf flag} = 1$}{
		\If{$c_{a_i} < t_{a_i}$}{
			$\tilde{a}_i \gets a_i$ \Comment{Set $\tilde{a}_i$ to $a_i$} \\
			$c_{a_i} \gets c_{a_i} +1$}
		\ElseIf{$\bar{c} < \bar{t}$}{
				$\tilde{a}_i \gets \phi$ \Comment{Set $\tilde{a}_i$ to $\phi$}\\
				$\bar{c} \gets \bar{c} + 1$}
			\Else{
				${\sf flag} \gets 0$\\
				$\tilde{a}_i \gets$ first $\alpha$  for which $c_{\alpha} < t_{\alpha}$\\ \Comment{Set $\mathcal{A} \ni \tilde{a}_i \neq  a_i$}\\
				Increase that $c_{\alpha}$ by $1$}
		}	
	\Else{
	$\tilde{a}_i \gets$ first $\alpha$  for which $c_{\alpha} < t_{\alpha}$\\ \Comment{Set $\mathcal{A} \ni \tilde{a}_i \neq  a_i$}\\
				Increase that $c_{\alpha}$ by $1$
		}
}
\end{algorithm}

For fixed inputs $(n, \mathcal{A}, \mP_{\!A}, \epsilon)$, let us write the output of Algorithm \ref{alg:type_mapping} for an input $a^n$ simply as $\tilde{a}^n$. 
Then $(\tilde{a}_1,\dots,\tilde{a}_n)$ contains exactly $t_{\alpha}\triangleq \lfloor n(1-\epsilon)\mP_{\!A}(\alpha) \rfloor$ instances of $\alpha$ for each $\alpha \in \mathcal{A}$, and  exactly $\bar{t}\triangleq n - \sum_{\alpha\in\mathcal{A}}t_{\alpha}$ instances of $\phi$. Processing the time slots from $i=1$ to $i=n$ sequentially (causally), the algorithm looks at $a_i$, and sets $\tilde{a}_i = a_i$ if the budget $t_{a_i}$ for $a_i$ has not yet been exhausted. Otherwise, it assigns $\tilde{a}_i=\phi$, if the budget $\bar{t}$ for $\phi$ is not exhausted. If both budgets are exhausted, then it changes strategy (and marks ${\sf flag} = 0$ to indicate the change) to instead set $\tilde{a}_i=\alpha \neq a_i$ for the first $\alpha$ whose budget is not yet exhausted, in order to have exactly $t_{\alpha}$ elements of $\alpha$ for every $\alpha \in \mathcal{A}$ in $\tilde{a}^n$, and thus exactly $\bar{t}$ elements of $\phi$ in $\tilde{a}^n$.

We illustrate two cases of this algorithm in Fig. \ref{fig:mapping}, one with returned ${\sf flag} = 1$ and the other one with returned ${\sf flag} = 0$.

\begin{figure}
\center
\begin{tikzpicture}
\def\rectwidth{0.5}    
\def\rectheight{0.5}   
\def\xgap{0.05}         

\definecolor{color0}{RGB}{178,243,178}
\definecolor{color1}{RGB}{248,207,245}
\definecolor{color2}{RGB}{127,127,127}
\definecolor{myred}{RGB}{243,92,109}
\definecolor{myblue}{RGB}{46,89,170}
\definecolor{mygreen}{RGB}{59,138,42}
\definecolor{myyellow}{RGB}{242,145,57}
\definecolor{mypurple}{RGB}{127,110,180}

\def\dataA{
0/0,
1/1,
0/0,
0/0,
1/1,
1/1,
1/1,
0/0,
1/1,
0/0
}
\def\dataAp{
0/0,
1/1,
0/0,
0/0,
1/1,
1/1,
2/\phi,
0/0,
2/\phi,
0/0
}
\def\dataB{
0/0,
1/1,
0/0,
0/0,
1/1,
1/1,
1/1,
1/1,
1/1,
0/0
}
\def\dataBp{
0/0,
1/1,
0/0,
0/0,
1/1,
1/1,
2/\phi,
2/\phi,
0/0,
0/0
}

\node[ ] at (2.8,1.45) {
\scriptsize {\color{black} $n=10$}, \colorbox{color0!50}{$t_{0}=5$},~
\colorbox{color1!50}{$t_{1}=3$},~
\colorbox{color2!30}{$\bar{t}=2$}
};

\node at (2.8,0.8) []{\footnotesize Algorithm \ref{alg:type_mapping} with returned ${\sf flag}=1$};
\begin{scope}[shift = {(0,0)}]
\node at (-0.5,0.3) {\footnotesize $a^n$};
\foreach \d/\a [count=\i] in \dataA {
  \pgfmathsetmacro{\x}{(\i-1)*(\rectwidth+\xgap)}
  \filldraw[color=color\d,draw=black] (\x,0) rectangle ++(\rectwidth,\rectheight);
  \node at (\x+0.5*\rectwidth,0.5*\rectheight) {\footnotesize $\a$};
}
\node at (-0.5,-0.5) {\tiny ${\sf flag} \gets 1$};
\draw[-{latex}, thick] (3.55,-0.2)--(3.55,-0.8);
\draw[-{latex}, thick] (4.65,-0.2)--(4.65,-0.8);
\end{scope}

\begin{scope}[shift = {(0,-1.5)}]
\node at (-0.5,0.3) {\footnotesize $\tilde{a}^n$};
\foreach \d/\a [count=\i] in \dataAp {
  \pgfmathsetmacro{\x}{(\i-1)*(\rectwidth+\xgap)}
  \filldraw[color=color\d,draw=black] (\x,0) rectangle ++(\rectwidth,\rectheight);
  \node at (\x+0.5*\rectwidth,0.5*\rectheight) {\footnotesize $\a$};
  }
\end{scope}
\node at (2.8,-2.2) []{\footnotesize Algorithm \ref{alg:type_mapping} with returned ${\sf flag} = 0$ };
\begin{scope}[shift = {(0,-3)}]
\node at (-0.5,0.3) {\small $a^n$};
\foreach \d/\a [count=\i] in \dataB {
  \pgfmathsetmacro{\x}{(\i-1)*(\rectwidth+\xgap)}
  \filldraw[color=color\d,draw=black] (\x,0) rectangle ++(\rectwidth,\rectheight);
  \node at (\x+0.5*\rectwidth,0.5*\rectheight) {\footnotesize $\a$};
  }
  \node at (-0.5,-0.5) {\tiny ${\sf flag} \gets 1$};
  \draw[-{latex}, thick] (3.55,-0.2)--(3.55,-0.8);
  \draw[-{latex}, thick] (4.1,-0.2)--(4.1,-0.8);
  \draw[-{latex}, thick, color=myred] (4.7,-0.2)--(4.7,-0.8) node [above right=0.1cm and 0cm] {\tiny {\color{myred} ${\sf flag} \gets 0$}};
\end{scope}

\begin{scope}[shift = {(0,-4.5)}]
\node at (-0.5,0.3) {\small $\tilde{a}^n$};
\foreach \d/\a [count=\i] in \dataBp {
  \pgfmathsetmacro{\x}{(\i-1)*(\rectwidth+\xgap)}
  \filldraw[color=color\d,draw=black] (\x,0) rectangle ++(\rectwidth,\rectheight);
  \node at (\x+0.5*\rectwidth,0.5*\rectheight) {\footnotesize $\a$};
  }
\end{scope}

\end{tikzpicture}
\caption{Running instances of Algorithm \ref{alg:type_mapping}. $\mathcal{A}= \{0,1\}$.}
\label{fig:mapping}
\end{figure}

Let us summarize the key properties. 

\begin{enumerate}
\item[P1.] The output type ${\bf N}_{\tilde{a}^n}$ is the same for all $a^n \in\mathcal{A}^n$. 
\item[P2.] Fixing $a^i$ fixes $\tilde{a}^i$ and therefore also fixes ${\bf N}_{\tilde{a}_{i+1}^n}=  {\bf N}_{\tilde{a}^n} - {\bf N}_{\tilde{a}^i}$, i.e., the type of $\tilde{a}_{i+1}^n$.
\end{enumerate}

\begin{remark} \label{rem:flag}
	A sufficient and necessary condition for returning ${\sf flag}=1$ is that ${\bf N}_{a^n}(\alpha) \geq t_{\alpha}, \forall \alpha \in \mathcal{A}$. 
	For the sufficiency, if ${\bf N}_{a^n}(\alpha) \geq t_{\alpha}, \forall \alpha \in \mathcal{A}$, then the condition in Line 9 will be violated exactly $\big({\bf N}_{a^n}(\alpha) - t_{\alpha} \big)$ times for each $\alpha \in \mathcal{A}$. For each time Line 9 is violated, the condition in Line 12 must be satisfied as $\sum_{\alpha}\big({\bf N}_{a^n}(\alpha)-t_{\alpha} \big) =\bar{t}$. This means that Line 16-Line 18 will not be executed, so ${\sf flag}=1$ is returned. For the necessity, if ${\bf N}_{a^n}(\alpha^*) < t_{\alpha^*}$ for some $\alpha^* \in \mathcal{A} $, then there exists $i\in [n]$, for which $\tilde{a}_i  = \alpha^* \neq  a_i$, as in the end there must be exactly $t_{\alpha^*}$ elements of $\alpha^*$ in $\tilde{a}^n$. This means Line 17 and Line 18 must be executed, and thus ${\sf flag}=0$ is returned.
\end{remark}

\subsection{Coding scheme}
We are now ready to construct the coding scheme. 
First, according to the definition of conditional distribution, any scheme $\mZ(x^n, \widehat{w} \mid w, s^n, y^n)$ can be factorized as
\begin{align} 
	\mZ(x^n\mid w, s^n, y^n) \times \mZ(\widehat{w} \mid x^n, w, s^n, y^n).\label{eq:factorZ}  
\end{align}
The first factor $\mZ(x^n\mid w, s^n, y^n)$ describes the  conditional distribution of the scheme generating the outputs $X^n$ given $W,S^n, Y^n$. The second factor $\mZ(\widehat{w} \mid x^n, w, s^n, y^n)$ describes the conditional distribution of the scheme generating the output $\widehat{W}$ given $W,S^n, Y^n$ and $X^n$.

For our achievability proof, we need to first specify $\mZ(x^n\mid w, s^n, y^n)$ and $\mZ(\widehat{w} \mid x^n, w, s^n, y^n)$ so that the resulting $\mZ(x^n, \widehat{w} \mid w, s^n, y^n)$ satisfies conditions \textit{C1},\textit{C2},\textit{C3}.
Our choice of coding scheme, which we will refer to as the \emph{authentication solution}, has the first factor
\begin{align} \label{eq:NSscheme_1}
\hspace{-0.2cm}	\mZ(x^n~|~ w, s^n, y^n) =  \prod_{i=1}^n \zeta_i(x_i~|~  s^i),
\end{align}
where for $i\in [n]$, $\zeta_i(x_i~|~  s^i)$ specifies the probability of outputting $X_i=x_i$ given the channel states up to the first $i$ channel uses, i.e., $s^i$.  The product form implies that $(X_1, X_2,\ldots, X_n)$ are mutually independent given $s^n$.

The authentication solution has the second factor
\begin{equation} \label{eq:NSscheme_2}
\begin{aligned}
	&\mZ(\widehat{w} \mid x^n, w, s^n, y^n)  \\
	&=\begin{cases}
	 	T(x^n,y^n,s^n), &  \widehat{w}=w \\
	 	\tfrac{1}{M-1} \big(1-T(x^n,y^n,s^n)\big) , &   \widehat{w}\neq w
	 	\end{cases},
\end{aligned}
\end{equation}
where $T\colon \mathcal{X}^n \times \mathcal{Y}^n \times \mathcal{S}^n \to \{0,\lambda\}$ is a binary function (to be specified) that maps every $(x^n, y^n, s^n)$ to either $0$ or some constant $\lambda \in (0,1]$, regardless of $(w, \widehat{w})$. Note that \eqref{eq:NSscheme_2} specifies a valid distribution for any $\lambda \in (0,1]$.  One could for now just think of $\lambda = 1$; the need for a value smaller than $1$ is to solve some rounding issue that will become clear later.
The distinctive features of the authentication solution are reflected in \eqref{eq:NSscheme_1} and \eqref{eq:NSscheme_2}. Firstly, \eqref{eq:NSscheme_1} means that the distribution of $X^n$ does not depend on $y^n$ or $w$, which suffices to satisfy \textit{C1}.
Secondly, given $(x^n, y^n, s^n, w)$, the distribution of $\widehat{W}$ depends on $w$ only through $\mathbb{I}[\widehat{W}=w]$,  regardless of how far  $\widehat{W}$ is from $w$. We point out that the intuition for the authentication solution comes from the simplification steps based on the idea of twirling used in \cite{cubitt2011zero, matthews2012linear, fawzi2024MAC, fawzi2024broadcast}.  

Intuitively, the authentication solution works as follows. At time slot $i$, the scheme $\mZ$ causally draws a random $X_i$ according to a distribution that depends only on $S^i$.
After $n$ time slots, with the receiver's input $Y^n$, the scheme `authenticates' whether $(X^n,Y^n, S^n)$ satisfy a certain relation (indicated by $T$). If the authentication is successful, then the scheme flips an (unfair) coin and outputs the correct message $\widehat{W}=W$ for the receiver with probability $\lambda$, or uniformly outputs some incorrect message $\widehat{W}\neq W$ with probability $1-\lambda$. Otherwise, the authentication fails, and the scheme uniformly outputs some incorrect message $\widehat{W}\neq W$. 

Note that for $i\in [0:n]$, corresponding to $S^n=s^n, Y^n=y^n, W=w, \widehat{W}=\widehat{w}$, and $X^i=x^i$ (we assume $X^0$ and $x^0$ are empty), we have
\begin{align}
	&\mZ(x^i, \widehat{w}\mid w,s^n,y^n) = \sum_{x_{i+1}^n} \mZ(x^n, \widehat{w}\mid w,s^n,y^n) \label{eq:expandZ}\\
	&= \prod_{j=1}^i \zeta_j(x_j\mid s^j) \times \sum_{x_{i+1}^n} \prod_{j=i+1}^n \zeta_j (x_j\mid s^j) \mZ(\widehat{w} \mid x^n,w, s^n,y^n)\\ 
	&= \prod_{j=1}^i \zeta_j(x_j\mid s^j)\times \mathbb{E}\left[\mZ(\widehat{w} \mid x^i, X_{i+1}^n,w, s^n,y^n)\right]\\
	&= \prod_{j=1}^i \zeta_j(x_j\mid s^j) \times 
	\begin{cases}
		\mathbb{E}[T(x^i,X_{i+1}^n, y^n, s^n)], & \widehat{w}=w\\
		\tfrac{1}{M-1}(1-\mathbb{E}[T(x^i,X_{i+1}^n, y^n, s^n)]), & \widehat{w}\not= w
	\end{cases} \label{eq:expectation_form}
\end{align}
where the expectations are taken with respect to  $X_{i+1}^n \sim \prod_{j=i+1}^n \zeta_j(x_j\mid s^j)$.
The form \eqref{eq:expectation_form} will be useful in the following analysis.

In the following we specify the authentication solution and verify that it also satisfies \textit{C2} and \textit{C3}. 

For $\epsilon \in (0,1)$, and $s^n\in \mathcal{S}^n$, let us run Algorithm \ref{alg:type_mapping} with inputs $(n,\mathcal{S}, \mP_{\!S}, s^n, \epsilon)$. Denote the output sequence associated with this $s^n$ as $\tilde{s}^n$, and denote the flag returned by the algorithm as ${\sf flag}[s^n]$.
We then specify the first factor, i.e., \eqref{eq:NSscheme_1}, of the authentication solution, by setting $\forall i\in [n], x_i \in \mathcal{X}$,
\begin{align} \label{eq:NSscheme_1_final}
\zeta_i(x_i~|~  s^i) &= \left\{\begin{array}{ll}
{\mP}_{X \mid S} (x_i \mid \tilde{s}_i),&\tilde{s}_i\neq \phi,\\
|\mathcal{X}|^{-1},&\tilde{s}_i= \phi.
\end{array}
\right.\\
&\triangleq {\mP}_{X \mid \tilde{S}}(x_i \mid \tilde{s}_i) 
\end{align}
This is feasible because due to Property P2 of Algorithm \ref{alg:type_mapping}, $\tilde{s}_i$ is determined by $s^i$, i.e., the channel state realizations up to the $i^{th}$ use of the channel, regardless of $s_{i+1}^n$.

Next we specify $T$ to specify the second factor, i.e., \eqref{eq:NSscheme_2} of the authentication solution.  
Given $s^n$ and its associated $\tilde{s}^n$, for $\sigma \in \mathcal{S}\cup \{\phi\}$, let $I_{\sigma} \triangleq \{i\in [n]\colon \tilde{s}_i = \sigma\}$ be the subset of indices $\{i\}$ for which $\tilde{s}_i=\sigma$. Also, define $n_\s \triangleq |I_{\sigma}|$ for $\sigma \in \mathcal{S}\cup \{\phi\}$. In the following we use bold letters to denote the corresponding length-$n$ sequence, e.g., ${\bf y} = y^n$. We also write ${\bf y}_{I} = (y_i)_{i\in I}$ for a subset $I\subseteq [n]$.

For each $y^n\in \mathcal{Y}^n$, we first map it to another sequence $\tilde{y}^n \in (\mathcal{Y} \times \{\phi\})^n$ by using Algorithm \ref{alg:type_mapping} in total $|\mathcal{S}|$ times. Specifically, for each $\sigma \in \mathcal{S}$, we run Algorithm \ref{alg:type_mapping} with inputs $(n_\sigma, \mathcal{Y}, \mP_{\!Y\mid S=\sigma}, {\bf y}_{I_{\sigma}}, \epsilon)$ and the output sequence associated with $\sigma$ is stored in $\tilde{\bf y}_{I_{\sigma}}$, i.e., the sequence $\tilde{y}^n$ at positions identified by $I_{\sigma}$. Here, for $\sigma\in \mathcal{S}$, $\mP_{\!Y\mid S=\sigma}$ is defined as the marginal distribution of $\mP_{\!XY\mid S=\s} = \mP_{\!X\mid S=\sigma}\mN_{Y\mid X,S=\sigma}$ for $Y$. Denote the flag returned by the algorithm associated with $\sigma$ as ${\sf flag}[y^n \mid S=\sigma]$. This defines $\tilde{y}_i$ for $i \in \cup_{\sigma\in \mathcal{S}} I_{\sigma}$. We further let $\tilde{y}_i = \phi$ for all $i\in I_{\phi}$. This completes the definition of $\tilde{y}^n$. 

In order to identify the positions at which the elements of $\tilde{\bf y}$ are not $\phi$, we further define, for $\s \in \mathcal{S}$,
\begin{align}
	\tilde{I}_{\sigma} \triangleq \{i\in [n]\colon \tilde{s}_i =\sigma, \tilde{y}_i \not= \phi\},
\end{align}
and define $\tilde{n}_{\s} \triangleq |\tilde{I}_{\s}|$.

\begin{remark} \label{rem:same_type}
According to Property P1 of Algorithm \ref{alg:type_mapping}, for each $\s\in \mathcal{S}\cup \{\phi\}$, $n_{\s}$ is invariant under changes of $s^n \in \mathcal{S}^n$, since the type of $\tilde{s}^n$ is invariant. Similarly, for each $\s \in \mathcal{S}$, $\tilde{n}_{\s}$ is invariant under changes of $y^n$. In addition, $\tilde{y}_i =\phi$ for all $i\in I_{\phi}$. It follows that the type of $\tilde{y}^n$ is invariant under changes of $(s^n,y^n) \in \mathcal{S}^n\times \mathcal{Y}^n$.
\end{remark}

Then we let, 
\begin{equation} \label{eq:def_T} 
\begin{aligned} 
	&T(x^n, y^n, s^n) \\
	&= \begin{cases}
			\lambda, & \mbox{if} ~({\bf x}_{\tilde{I}_\s}, \tilde{\bf y}_{\tilde{I}_\s}) \in \mathcal{T}_{\epsilon}^{(\tilde{n}_\s)}(\mP_{\!XY\mid S=\s}), \forall \s\in \mathcal{S}, \\
			0, & \mbox{otherwise.}
		\end{cases}	
\end{aligned}
\end{equation}
Although the RHS of \eqref{eq:def_T} does not explicitly contain $s^n$, we should point out that $s^n$ determines $I_{\sigma}, \forall \sigma \in \mathcal{S} \cup \{\phi\}$, and also affects $\tilde{\bf y}$ and $\tilde{I}_{\sigma}$ accordingly.

Let us show that \textit{C3} is satisfied by the design, i.e., $\mZ(x^i, \widehat{w} \mid w, s^n, y^n)$ is invariant under changes of $s_{i+1}^n$. By \eqref{eq:expectation_form} it suffices to show that $\mathbb{E}\big[ T(x^i,X_{i+1}^n, y^n, s^n) \big]$ is  invariant under changes of $s_{i+1}^n$.

To clarify notations, given $(s^i, y^i)\in \mathcal{S}^i \times \mathcal{Y}^i$,  two realizations in $\mathcal{S}^{n-i}$, say $s_{i+1}^n$ and ${s'}_{i+1}^n$, and two realizations of $\mathcal{Y}^{n-i}$, say $y_{i+1}^n$ and ${y'}_{i+1}^n$, let $s^n \triangleq (s^i, s_{i+1}^n), {s'}^n \triangleq (s^i, {s'}_{i+1}^n), y^n \triangleq (y^i, y_{i+1}^n), {y'}^n \triangleq (y^i, {y'}_{i+1}^n)$. 
Let $\tilde{s}^n$, $\tilde{y}^n$ denote the corresponding sequences mapped (according to the description above) from $s^n$ and (then) $y^n$, respectively. Let $\tilde{s'}^n$, $\tilde{y'}^n$ denote the corresponding sequences mapped from ${s'}^n$ and (then) ${y'}^n$, respectively. Then we need to show that,
\begin{align} \label{eq:same_expectation}
	\mathbb{E}\big[ T(x^i,X_{i+1}^n, y^n, s^n) \big] = \mathbb{E}\big[ T(x^i, {X'}_{i+1}^n, {y'}^n, {s'}^n) \big]
\end{align}
where $X_{i+1}^n \sim \prod_{j=i+1}^n \mP_{\!X\mid \tilde{S}}(x_j\mid \tilde{s}_j)$ and ${X'}_{i+1}^n \sim \prod_{j=i+1}^n \mP_{\!X\mid \tilde{S}}(x_j\mid \tilde{s}'_j)$. In fact, for \textit{C3} it suffices to consider ${y'}^n=y^n$. The reason we make it general is to also facilitate the proof for \textit{C2} later.
 
According to Property P2 of Algorithm \ref{alg:type_mapping},  $(\tilde{s}^i, \tilde{y}^i) = (\tilde{s'}^i, \tilde{y'}^i)$, and thus  $({\tilde{s}}_{i+1}^n, \tilde{y}_{i+1}^n)$ and $({\tilde{s'}}_{i+1}^n, \tilde{y'}_{i+1}^n)$ have the same (joint) type (see Remark \ref{rem:same_type}). This means that there is a permutation $\pi\colon [i+1:n] \to [i+1:n]$ of the indices $i+1,\ldots, n$ such that $(\tilde{s}'_k, \tilde{y}'_k) = (\tilde{s}_{\pi(k)}, \tilde{y}_{\pi(k)})$ for $k\in [i+1:n]$.

Now for a fixed  $x^i$, $y^n,s^n$, consider those $x_{i+1}^n \in \mathcal{X}^{n-i}$ for which $T(x^n, y^n, s^n) = \lambda$, i.e., the subset
\begin{align}
	\mathfrak{X}_{x^i,y^n,s^n} \triangleq \{x_{i+1}^n \colon T(x^n,y^n,s^n) = \lambda\}.
\end{align}
Since $(\tilde{s}'_k, \tilde{y}'_k) = (\tilde{s}_{\pi(k)}, \tilde{y}_{\pi(k)})$ for $k\in [i+1:n]$, those $x_{i+1}^n \in \mathcal{X}^{n-i}$ for which $T(x^n, {y'}^n, {s'}^n) = \lambda$ constitute the subset
\begin{align}
	&\mathfrak{X}_{x^i,{y'}^n,{s'}^n} \triangleq \{x_{i+1}^n \colon T(x^n,{y'}^n,{s'}^n) = \lambda\} \\
	&= \big\{{x'}_{i+1}^n \colon x_k' = x_{\pi(k)}, k\in[i+1:n],x_{i+1}^n \in \mathfrak{X}_{x^i,y^n,s^n} \big\},\notag
\end{align}
since $(X'_{i+1},\ldots, X'_{n})$ is drawn according to the same distribution as that for $(X_{\pi(i+1)},\ldots, X_{\pi(n)})$, and  $(\tilde{y}'_{i+1},\ldots, \tilde{y}'_{n}) = (\tilde{y}_{\pi(i+1)},\ldots, \tilde{y}_{\pi(n)})$. Therefore, the probability $\Pr\big(T((x^i, X_{i+1}^n),y^n, s^n)=\lambda \big)$ is the same as the probability  $\Pr\big(T((x^i, {X'}_{i+1}^n),{y'}^n, {s'}^n)=\lambda \big)$. It then follows that \eqref{eq:same_expectation} holds. Therefore, \textit {C3} is satisfied.

It remains to satisfy \textit{C2}, i.e., $\mZ(\widehat{w}\mid w,s^n,y^n)$ should be invariant under changes of $(w, s^n)$. Note that $\mZ(\widehat{w}\mid w,s^n,y^n)$ corresponds to the LHS of \eqref{eq:expandZ} when we set $i=0$.
In \eqref{eq:expectation_form} note that as a consequence of the above analysis, $\mathbb{E}\big[ T(X^n, y^n, s^n) \big]$ is invariant under the changes of $(s^n,y^n)$, since the resulting  $(\tilde{s}^n, \tilde{y}^n)$ will only have one type. However, \eqref{eq:expectation_form} may depend on $w$ because of its two forms depending on whether $\widehat{w}=w$ or $\widehat{w}\neq w$, whereas \textit{C2} requires that the value should not depend on $w$. To  remove the dependence on $w$,  it would be sufficient to have
\begin{equation}
	\mathbb{E}\big[ T(X^n, y^n, s^n) \big] = \tfrac{1}{M}, ~~~~ \forall y^n\in \mathcal{Y}^n, s^n \in \mathcal{S}^n, \label{eq:C1_authentication} 
\end{equation}
because this ensures that $\mathbb{E}\big[ T(X^n, y^n, s^n) \big]=\frac{1}{M-1}(1-\mathbb{E}\big[ T(X^n, y^n, s^n) \big])=1/M$.

This can be done by adjusting the value of the integer $M$ and $\lambda\in (0,1]$. 
Specifically, for $s^n \in \mathcal{S}^n$ and $y^n \in \mathcal{Y}^n$, 
\begin{equation}\label{eq:C1_intermediate}
\begin{aligned}
	&\mathbb{E}\big[ T(X^n, y^n, s^n) \big] \\
	&= \lambda \times \underbrace{\Pr\big( ({\bf X}_{\tilde{I}_\s}, \tilde{\bf y}_{\tilde{I}_\s}) \in \mathcal{T}_{\epsilon}^{(\tilde{n}_\s)}(\mP_{\!XY\mid S=\s}),\forall \s\in \mathcal{S} \big)}_{\triangleq (\mu)^{-1}} 
\end{aligned}
\end{equation} 
where $\mu$ is a constant for all $s^n \in \mathcal{S}^n, y^n \in \mathcal{Y}^n$, and  ${\bf X} = X^n \sim \prod_{j=1}^n \mP_{\!X\mid \tilde{S}}(x_j\mid \tilde{s}_j)$.  
Now, to satisfy \eqref{eq:C1_authentication}, we set
\begin{align}
M = \lceil \mu \rceil, ~~
	\lambda = \frac{\mu}{M}.\label{eq:deflambda}
\end{align}
Thus, $\mZ(\widehat{w} \mid w, s^n, y^n)=\tfrac{1}{M}$ for all $w,\widehat{w}, s^n, y^n$, and \textit{C2} is satisfied. A valid coding scheme $\mZ$ is now fully specified. We illustrate the construction in Fig. \ref{fig:authentication}.

\begin{figure}[t]
\centering
\begin{tikzpicture}
\def\w{1}
\def\h{0.3}
\def\hh{0.9}
\fill[gray!10] (-0.5*\w,0) -- (1*\w,0) -- (1*\w,-\h) -- (2*\w,-\h) -- (2*\w,-2*\h) -- (2*\w+0.25,-2*\h) decorate[decoration = {zigzag, amplitude=2pt, segment length=6pt}] {--  (2*\w+0.75,-2*\h)}-- (3*\w,-2*\h) -- (3*\w,-3*\h) -- (4*\w,-3*\h) -- (4*\w, 0) -- (11*\w, 0) -- (11*\w, -12*\h-\hh) -- (2*\w+0.75, -12*\h-\hh) decorate[decoration = {zigzag, amplitude=2pt, segment length=6pt}] {--  (2*\w+0.25,-12*\h-\hh)} -- (-0.5*\w, -12*\h-\hh) -- (-0.5*\w,0);
\draw[thick] (-0.5*\w,0) -- (1*\w,0) -- (1*\w,-\h) -- (2*\w,-\h) -- (2*\w,-2*\h) -- (2*\w+0.25,-2*\h) decorate[decoration = {zigzag, amplitude=2pt, segment length=6pt}] {--  (2*\w+0.75,-2*\h)}-- (3*\w,-2*\h) -- (3*\w,-3*\h) -- (4*\w,-3*\h) -- (4*\w, 0) -- (11*\w, 0) -- (11*\w, -12*\h-\hh) -- (2*\w+0.75, -12*\h-\hh) decorate[decoration = {zigzag, amplitude=2pt, segment length=6pt}] {--  (2*\w+0.25,-12*\h-\hh)} -- (-0.5*\w, -12*\h-\hh) -- (-0.5*\w,0);

\node (w) at (-0.2,0.5) [] {$w$}; \draw[thick, -{latex}] (w)-- ($(w.south)+(0,-0.3)$);
\node (s1) at (0.5,0.5) [] {$s_1$}; \draw[thick, -{latex}] (s1)-- ($(s1.south)+(0,-1.75)$);  \draw[thick, -{latex}] ($(s1.south)+(0,-2.4)$) -- ($(s1.south)+(0,-3.25)$) node[pos=0.4, left=-0.1] {\small $\tilde{s}_1$};
\node (s2) at (1.5,0.5) [] {$s_2$}; \draw[thick, -{latex}] (s2)-- ($(s2.south)+(0,-1.75)$); \draw[thick, -{latex}] ($(s2.south)+(0,-2.4)$) -- ($(s2.south)+(0,-3.25)$) node[pos=0.4, right=-0.05] {\small $\tilde{s}_2$};
\node (sn) at (3.5,0.5) [] {$s_n$}; \draw[thick, -{latex}] (sn)-- ($(sn.south)+(0,-1.75)$); \draw[thick, -{latex}] ($(sn.south)+(0,-2.4)$) -- ($(sn.south)+(0,-3.25)$) node[pos=0.4, left=-0.1] {\small $\tilde{s}_n$};

\node[rectangle, draw,  thick, minimum width = 4cm, minimum height=0.6cm, fill=black!10] at (2,-1.8) (Alg1) {Alg. \ref{alg:type_mapping}};

\node[rectangle, draw,  thick, minimum height=0.6cm, fill=black!10] at (0.4,-3.3) {\footnotesize $\mP_{\!X\mid \tilde{S}}$}; \draw[thick, -{latex}] ($(s1.south)+(0,-3.85)$) -- ($(s1.south)+(0,-5.1)$); \node at ($(s1.south)+(0,-5.3)$) {\small $X_1$};
\node[rectangle, draw,  thick, minimum height=0.6cm, fill=black!10] at (1.6,-3.3) {\footnotesize $\mP_{\!X\mid \tilde{S}}$}; \draw[thick, -{latex}] ($(s2.south)+(0,-3.85)$) -- ($(s2.south)+(0,-5.1)$); \node at ($(s2.south)+(0,-5.3)$) {\small $X_2$};
\node[rectangle, draw,  thick, minimum height=0.6cm, fill=black!10] at (3.5,-3.3) {\footnotesize $\mP_{\!X\mid \tilde{S}}$}; \draw[thick, -{latex}] ($(sn.south)+(0,-3.85)$) -- ($(sn.south)+(0,-5.1)$); \node at ($(sn.south)+(0,-5.3)$) {\small $X_n$};

\node at (7.7,-5) (y) {$y^n$}; \draw[thick, -{latex}] ($(y.north)+(0,0)$) -- ($(y.north)+(0,0.5)$); 
\node[rectangle, draw,  thick, minimum height=0.6cm, minimum width=5.2cm, fill=black!10, above right = 0.5cm and -2.8cm of y] {\footnotesize Permute based on $(I_{\sigma})_{\sigma \in \mathcal{S} \cup \{\phi\}}$};

\draw [thick, -{latex}] ($(y.north)+(-1.8,1.15)$) -- ($(y.north)+(-1.8,1.65)$) node[pos=0.4, left=-0.1] {\small ${\bf y}_{I_1}$};
\draw [thick, -{latex}] ($(y.north)+(0,1.15)$) -- ($(y.north)+(0,1.65)$) node[pos=0.4, left=-0.1] {\small ${\bf y}_{I_2}$};
\draw [thick, -{latex}] ($(y.north)+(2.2,1.15)$) -- ($(y.north)+(2.2,1.65)$) node[pos=0.4, left=-0.1] {\small ${\bf y}_{I_{|\mathcal{S}|}}$};

\node[rectangle, draw,  thick,  minimum height=0.6cm, fill=black!10] at (5.8,-2.75)  {\small Alg. \ref{alg:type_mapping}}; \draw [thick, -{latex}] ($(y.north)+(-1.8,2.25)$) -- ($(y.north)+(-1.8,2.8)$) node[pos=0.4, left=-0.1] {\small $\tilde{\bf y}_{I_1}$};
\node[rectangle, draw,  thick,  minimum height=0.6cm, fill=black!10] at (7.7,-2.75)  {\small Alg. \ref{alg:type_mapping}}; \draw [thick, -{latex}] ($(y.north)+(0,2.25)$) -- ($(y.north)+(0,2.8)$) node[pos=0.4, left=-0.1] {\small $\tilde{\bf y}_{I_2}$};
\node[rectangle, draw,  thick,  minimum height=0.6cm, fill=black!10] at (9.8,-2.75)  {\small Alg. \ref{alg:type_mapping}}; \draw [thick, -{latex}] ($(y.north)+(2.2,2.25)$) -- ($(y.north)+(2.2,2.8)$) node[pos=0.4, left=-0.1] {\small $\tilde{\bf y}_{I_{|\mathcal{S}|}}$};

\node (J) [rectangle, rounded corners, fill=gray!20, thick, draw] at (7.8,-1.55) {\scriptsize Is $\big({\bf X}_{\tilde{I}_\sigma}, \tilde{\bf y}_{\tilde{I}_\sigma}\big) \in \mathcal{T}^{(\tilde{n}_\s)}_{\epsilon}, \forall \sigma \in \mathcal{S}$ and $\Lambda=1$?};

\draw [thick, -{latex}] ($(J.north)+(-2,0)$) -- ($(J.north)+(-2,0.5)$) node[pos=0.4, left=0] {\scriptsize Yes}; \node at ($(J.north)+(-2,0.75)$) {\scriptsize $\widehat{W}=w$};

\draw [thick, -{latex}] ($(J.north)+(1,0)$) -- ($(J.north)+(1,0.5)$) node[pos=0.4, right=0] {\scriptsize No}; \node at ($(J.north)+(1.5,0.75)$) {\scriptsize $\widehat{W} \sim {\rm Unif}([M]\setminus \{w\})$};

\draw [thick, -{latex}] (7.5,0) -- (7.5,0.3);
\node at (7.52,0.55) {\footnotesize $\widehat{W}$};

\node at (2.5,0.5) {\footnotesize $\cdots$};
\node at (2.55,-3.3) {\footnotesize $\cdots$};
\node at (2.5,-5) {\footnotesize $\cdots$};
\node at (8.8,-2.75) {\footnotesize $\cdots$};

\end{tikzpicture}
\caption{Illustration of the scheme construction. We assume the channel state alphabet $\mathcal{S} = \{1,2,\ldots, |\mathcal{S}|\}$. $\Lambda$ is an independent classical random variable internally generated by the scheme with $\Pr(\Lambda =1) = \lambda$ and $\Pr(\lambda =0) = 1-\lambda$.}
\label{fig:authentication}
\end{figure}

\subsection{Analysis}
First, consider the expectation in \eqref{eq:C1_intermediate}.
Fixing $s^n \in \mathcal{S}^n, y^n \in \mathcal{Y}^n$, for each $\sigma \in \mathcal{S}$, ${\bf X}_{\tilde{I}_{\s}}$ are generated i.i.d. according to $\mP_{\!X\mid S=\sigma}$. According to the Joint Typicality Lemma \cite{NIT}, it holds for any  sequence in $\mathcal{Y}^{\tilde{n}_{\s}}$ (and thus for $\tilde{\bf y}_{\tilde{I}_{\s}}$) that
\begin{align}
	\Pr\big( ({\bf X}_{\tilde{I}_{\s}}, \tilde{\bf y}_{\tilde{I}_{\s}}) \in \mathcal{T}_{\epsilon}^{(\tilde{n}_{\s})} \big) \leq 2^{-\tilde{n}_{\s}(I(X;Y\mid S=\s)-\delta(\epsilon))},\label{eq:MIboundforeachsigma}
\end{align}
where $\delta(\epsilon) \to 0$ as $\epsilon \to 0$. 
According to Algorithm \ref{alg:type_mapping}, for $\sigma \in \mathcal{S}$,
\begin{align}
	\tilde{n}_{\sigma} &= \sum_{y\in \mathcal{Y}} \lfloor n_{\s}   (1-\epsilon) \mP_{\!Y\mid S=\s}(y) \rfloor \\
	&\geq \sum_{y\in \mathcal{Y}} \big( n_{\s}   (1-\epsilon) \mP_{\!Y\mid S=\s}(y)   - 1 \big)\\
	&= n_{\s}  (1-\epsilon)- o(n)
\end{align}
Meanwhile, 
\begin{align}
	n_\s = \lfloor n(1-\epsilon) \mP_{\!S}(\s) \rfloor \geq n(1-\epsilon) \mP_{\!S}(\s)-o(n).
\end{align}
It follows that
\begin{align} \label{eq:relation_tilde}
	\tilde{n}_\s \geq   n(1-\epsilon)^2 \mP_{\!S}(\sigma) - o(n).
\end{align}
Consider all $\sigma \in \mathcal{S}$ and note that $X_1,X_2,\ldots, X_n$ are mutually independent given $s^n$. We have
\begin{align}
	M^{-1}&\stackrel{\eqref{eq:deflambda}}{\leq} (\mu)^{-1} =\prod_{\s\in\mathcal{S}}\Pr\big( ({\bf X}_{\tilde{I}_{\s}}, \tilde{\bf y}_{\tilde{I}_{\s}}) \in \mathcal{T}_{\epsilon}^{(\tilde{n}_{\s})} \big)\\
	&\stackrel{\eqref{eq:MIboundforeachsigma}}{\leq}  \prod_{\s\in \mathcal{S}} 2^{-\tilde{n}_{\s} (I(X;Y\mid S=\s)-\delta(\epsilon))}\\
	&\stackrel{\eqref{eq:relation_tilde}}{\leq} \prod_{\s\in \mathcal{S}} 2^{-n(1-\epsilon)^2\mP_{\!S}(\sigma)(I(X;Y\mid S=\s) - \delta(\epsilon) )   +o(n)} \\
	& = 2^{-n (1-\epsilon)^2 (I(X;Y\mid S) -\delta(\epsilon) )+o(n)}
\end{align}
Therefore, asymptotically the rate of the coding scheme is bounded below as,
\begin{align}
	\lim_{n\to \infty} \frac{\log M}{n} \geq  (1-\epsilon)^2 \big(I(X;Y\mid S)-\delta(\epsilon)
	\big)
\end{align} 
which approaches $I(X;Y\mid S)$ for $\epsilon\rightarrow 0$.

Next we analyze the probability of success associated with the scheme. According to \eqref{eq:prob_succ_NS}, \eqref{eq:factorZ}, \eqref{eq:NSscheme_1}, \eqref{eq:NSscheme_1_final} and \eqref{eq:NSscheme_2}, this is
\begin{align}
	 \eta({\sf Z}) &= \sum_{x^n, y^n, s^n}\prod_{i=1}^n \mP_{\!S}(s_i) \cdot \mN (y_i\mid x_i, s_i) \cdot \mP_{\!X\mid \tilde{S}}(x_i\mid \tilde{s}_i) \cdot T(x^n,y^n,s^n) \\
	 &= \mathbb{E}\big[ T(X^n, Y^n, S^n) \big] \label{eq:prob_success_interm}
\end{align}
where for the expectation in \eqref{eq:prob_success_interm}, $(X^n,Y^n,S^n) \sim \prod_{i=1}^n \mP_{\!S}(s_i) \mP_{\!X\mid \tilde{S}}(x_i \mid \tilde{s}_i) \mN(y_i\mid x_i,s_i)$.

It is worthwhile to note the distinction between the $Y^n$ in \eqref{eq:prob_success_interm} and the $y^n$ in  \eqref{eq:C1_authentication}.  The random variable $Y^n$ is the output of the channel $\mN$ in state $S^n=s^n$ with input $X^n$. In particular $Y^n$ and $X^n$ are not independent. In contrast $y^n$ in \eqref{eq:C1_authentication} is an arbitrary sequence. Thus, even though for every $y^n$, we have $\mathbb{E}\big[ T(X^n, y^n, s^n) \big] = \tfrac{1}{M}$ according to \eqref{eq:C1_authentication}, one should not conclude that  $\mathbb{E}\big[ T(X^n, Y^n, S^n) \big] = \tfrac{1}{M}$ in \eqref{eq:prob_success_interm}. In fact $\eta(\mZ)=\mathbb{E}\big[ T(X^n, Y^n, S^n) \big]$ approaches $1$ as $n\to \infty$, as we prove next.

To continue, since $S^n$ is random, the output of Algorithm \ref{alg:type_mapping} for input $S^n$ is a random sequence $\tilde{S}^n$, together with a random  ${\sf flag}[S^n]$. Moreover, $Y^n$ is random, and the output of Algorithm \ref{alg:type_mapping} with input ${\bf Y}_{I_{\s}}$ for each $\sigma \in \mathcal{S}$ is also random, denoted as $\widetilde{\bf Y}_{I_{\sigma}}$, together with a random ${\sf flag}[Y^n\mid S=\sigma]$. Define a random variable $F$ to indicate whether all $(1+|\mathcal{S}|)$ flags are $1$, as
\begin{align}
	F \triangleq \mathbb{I}\big[{\sf flag}[S^n]=1\big] \times \prod_{\s \in \mathcal{S}}  \mathbb{I}\big[ {\sf flag}[Y^n\mid S=\sigma] =1 \big].
\end{align}
We continue from \eqref{eq:prob_success_interm} as,
\begin{align}
	 \eta(\sf Z)&= \lambda \times \Pr\big( ({\bf X}_{\tilde{I}_{\s}}, \widetilde{\bf Y}_{\tilde{I}_{\s}}) \in \mathcal{T}_{\epsilon}^{(\tilde{n}_\sigma)}, \forall \s\in \mathcal{S} \big) \\
	&\geq  \lambda\times \Pr(F = 1)\times  \notag \\
	&\qquad \Pr\big( ({\bf X}_{\tilde{I}_{\s}}, \widetilde{\bf Y}_{\tilde{I}_{\s}}) \in \mathcal{T}_{\epsilon}^{(\tilde{n}_{\s})}, \forall \s\in \mathcal{S} \mid F = 1\big) \label{eq:prob_success_interm2}
\end{align}
According to the definition of strong typical sets in \eqref{eq:def_strong_typical_set}, if $S^n \in \mathcal{T}^{(n)}_{\epsilon}(\mP_{\!S})$, then ${\bf N}_{S^n}(\sigma) \geq n(1-\epsilon)\mP_{\!S}(\sigma) \geq \lfloor n(1-\epsilon)\mP_{\!S}(\sigma) \rfloor$, and  Remark \ref{rem:flag} then says that ${\sf flag}[S^n]=1$. Since $\Pr\big(S^n \in \mathcal{T}^{(n)}_{\epsilon}\big) \to 1$ as $n\to \infty$, we have
\begin{align} \label{eq:flag1}
	\Pr({\sf flag}[S^n]=1) \to 1
\end{align}
as $n\to \infty$. 
Then conditioned on ${\sf flag}[S^n]=1$, we have that $\tilde{\bf S}_{I_{\s}} = {\bf S}_{I_{\s}}$ for every $\s\in \mathcal{S}$, meaning that $\tilde{S}^n$ is unchanged from $S^n$ at the positions indexed by $I_{\sigma}$ for every $\s\in \mathcal{S}$. This means that conditioned on ${\sf flag}[S^n]=1$ and $S^n=s^n$, for each $\sigma \in \mathcal{S}$, at the positions $j \in I_{\sigma}$  we have $(X_j, Y_j) \sim  \mP_{\!X\mid S=\sigma} \mN_{Y\mid X, S=\sigma} $, and thus $Y_j \sim \mP_{\!Y\mid S=\sigma}$. Moreover, if ${\bf Y}_{I_{\s}} \in \mathcal{T}_{\epsilon}^{(n_{\s})}(\mP_{\!Y\mid S=\s})$, then we similarly have ${\sf flag}[Y^n \mid S=\s] = 1$ for $\s\in \mathcal{S}$. Since $\Pr\big({\bf Y}_{I_{\s}} \in \mathcal{T}_{\epsilon}^{(n_{\s})}\big) \to 1$ as $n\to \infty$, we have
\begin{align} \label{eq:flag2}
	\Pr\big( {\sf flag}[Y^n\mid S=\s]  \mid {\sf flag}[S^n] \big) \to 1, \forall \s\in \mathcal{S}.
\end{align}
Combining \eqref{eq:flag1} and \eqref{eq:flag2}, it follows that 
\begin{align}
	\Pr(F=1) \to 1
\end{align}
as $n\to \infty$. 

Note that conditioned on $F=1$, $\widetilde{\bf Y}_{\tilde{I}_{\s}}={\bf Y}_{\tilde{I}_{\s}}$ for all $\s\in \mathcal{S}$. We continue from \eqref{eq:prob_success_interm2} that
\begin{align}
	&\eqref{eq:prob_success_interm2} = \lambda\times \Pr(F = 1)\times  \notag \\
	&\qquad \Pr\big( ({\bf X}_{\tilde{I}_{\s}}, {\bf Y}_{\tilde{I}_{\s}}) \in \mathcal{T}_{\epsilon}^{(\tilde{n}_{\s})}, \forall \s\in \mathcal{S} \mid F = 1\big) 
\end{align}
and the last factor approaches $1$ as $\tilde{n}_{\s} \to 1$ for all $\sigma \in \mathcal{S}$. This is guaranteed if $n \to \infty$ due to \eqref{eq:relation_tilde}. We conclude that $\eta(\mZ) \to 1$ as $n\to \infty$. \hfill \qed

\section{Conclusion}
In this work, we formalize and study the NS-assisted capacity of a discrete memoryless classical channel with causal channel state information at the transmitter (CSIT).
We show that when NS assistance is available, the capacity with causal CSIT coincides with that with non-causal CSIT, and that both are equal to the classical capacity of the corresponding channel with channel state information also available at the receiver (CSIR). However, unlike the non-causal CSIT case where a stronger equivalence at the level of probability of success has been previously established, we demonstrate with a counterexample that with causal CSIT that stronger equivalence does not hold. Separately, we show that NS assistance, feedback, and \emph{strictly} causal CSIT (i.e., transmitter knows past channel states but not the present or future states), even acting collectively, cannot increase the capacity of a discrete memoryless channel.

\appendix

\section{Proof of Theorem \ref{thm:Prob}} \label{proof:Prob}
Recall that $(\mN, \mP_{\!S})$ is the $Z_0/Z_1$ channel in Definition \ref{def:Z0Z1}. Let us first show that
\begin{equation*}
	{\eta}^{\C,\ca}_{\opt, M=2,n=2}(\mN^{\CSIR},\mP_{\!S}) \geq 7/8,
\end{equation*}
i.e., when sending a one-bit message over $2$ channel uses, the optimal probability of success for \emph{classical} coding schemes with causal CSIT and CSIR is at least $7/8$. For each channel use, since both the transmitter and the receiver know the current channel state, they can convert the channel always to $Z_0$ (by switching $1$ and $0$ at both the input and output $Z_1$ is transformed into $Z_0$). Therefore, we can design coding schemes assuming they have 2 uses of  the channel $Z_0$. We then adopt the following classical strategy. For encoding, the encoder maps $W=0$ to $(X_1,X_2) = (0,0)$ and $W=1$ to $(X_1,X_2)=(1,1)$. For decoding, the decoder maps $(Y_1,Y_2) = (0,0)$ to $\widehat{W}=0$ and $(Y_1,Y_2) \in \{(0,1),(1,0),(1,1)\}$ to $1$. Then conditioned on $W=0$, we always have $\widehat{W}=0=W$, and conditioned on $W=1$, we have $\Pr(\widehat{W}=1 \mid W=1) = \Pr((Y_1,Y_2) \not=(0,0) \mid (X_1,X_2) = (1,1)) = 3/4$. Thus, the probability of success for this strategy is $\frac{1}{2}(1+\frac{3}{4}) = 7/8$.

Next, we show that ${\eta}^{\NS,\ca}_{\opt,M=2,n=2}(\mN, \mP_{\!S}) \leq 13/16$, i.e., for the same $M,n$ as above, the optimal probability of success for NS-assisted coding schemes with causal CSIT cannot be more than $13/16$. 
We first derive a linear program for computing ${\eta}_{\opt,M,n}^{\NS,\ca}(\mN, \mP_{\!S})$ for a general channel with state.
Without loss of generality, say $0\in \mathcal{S}\cap \mathcal{Y}$. For $n\in \mathbb{N}$ let $0^n$ denote the all-$0$ sequence of length $n$. By writing $\mZ(x^n,\widehat{w}\mid w,s^n,y^n) = z_{x^n,\widehat{w}\mid w,s^n,y^n}$, we obtain that ${\eta}^{\NS,\ca}_{\opt,M,n}(\mN, \mP_{\!S})$ is the solution to the following linear program (LP1), with the set of variables $\{z_{x^n,\widehat{w}\mid w,s^n,y^n} \colon x^n\in \mathcal{X}^n, \widehat{w} \in [M], w \in [M], s^n \in \mathcal{S}^n, y^n \in \mathcal{Y}^n\}$,
\begin{align}
\mbox{maximize}\quad & \frac{1}{M}\sum_{w,s^n,x^n,y^n} z_{x^n,w\mid w,s^n,y^n} \cdot \mP_{\!S}^{\otimes n}(s^n) \cdot \mN^{\otimes n}(y^n\mid x^n,s^n)\\
\mbox{s.t.}\quad
& z_{x^n, \widehat{w}\mid w,s^n,y^n}  \geq 0, ~\forall (x^n,\widehat{w}, w,s^n, y^n) \label{eq:LP_norm1} \\
& \sum_{x^n,\widehat{w}} z_{x^n, \widehat{w}\mid w,s^n,y^n} = 1, ~\forall (w,s^n,y^n) \label{eq:LP_norm2} \\
&\sum_{\widehat{w}} z_{x^n, \widehat{w}\mid w,s^n,y^n} = \sum_{\widehat{w}} z_{x^n, \widehat{w}\mid w,s^n,y^n=0^n}, \notag \\& \qquad \forall (x^n,w,s^n,y^n) \label{eq:LP_NS1} \\
&\sum_{x^n} z_{x^n, \widehat{w}\mid w,s^n,y^n} = \sum_{x^n} z_{x^n, \widehat{w}\mid w=1,s^n=0^n,y^n}, \notag \\& \qquad \forall (\widehat{w},w, s^n ,y^n) \label{eq:LP_NS2} \\
&\sum_{x_{i+1}^n} z_{x^n, \widehat{w}\mid w,(s^i,s_{i+1}^n),y^n} = \sum_{x_{i+1}^n} z_{x^n, \widehat{w}\mid w,(s^i, s_{i+1}^n = 0^{n-i}),y^n},  \notag \\
&\qquad \forall i \in [n-1], \forall (x^i, \widehat{w}, w,s^n, y^n) \label{eq:LP_causal}
\end{align}
Conditions \eqref{eq:LP_norm1} and \eqref{eq:LP_norm2} state that $\mZ(x^n,\widehat{w}\mid w,s^n,y^n) = z_{x^n,\widehat{w}\mid w,s^n,y^n}$ is a valid conditional distribution, whereas Conditions \eqref{eq:LP_NS1}, \eqref{eq:LP_NS2} and \eqref{eq:LP_causal} correspond to the non-signaling and the causality conditions \textit{C1}--\textit{C3}.

Similar to the simplification steps used in \cite{matthews2012linear, fawzi2024MAC, fawzi2024broadcast}, we can write the above linear program as another smaller linear program (LP2) with a set of fewer variables $\{r_{x^n,y^n,s^n}\colon x^n\in \mathcal{X}^n,y^n\in \mathcal{Y}^n, s^n \in \mathcal{S}^n\} \cup \{q_{x^n\mid s^n}\colon x^n \in \mathcal{X}^n, s^n \in \mathcal{S}^n\}$. The simplified linear program is as follows.
\begin{align}
\mbox{maximize}\quad &  \sum_{x^n,y^n,s^n} r_{x^n,y^n,s^n}\cdot \mP_{\!S}^{\otimes n}(s^n) \cdot \mN^{\otimes n}(y^n\mid x^n,s^n) \\
\mbox{s.t.}\quad
& r_{x^n,y^n,s^n}  \geq 0, ~\forall (x^n, y^n, s^n) \label{eq:LP2_norm1} \\
& \sum_{x^n} r_{x^n,y^n,s^n} = \tfrac{1}{M}, ~\forall (s^n,y^n) \label{eq:LP2_norm2} \\
& \sum_{x^n} q_{x^n\mid s^n} = 1, ~\forall s^n \label{eq:LP2_norm3} \\
& r_{x^n,y^n,s^n} \leq q_{x^n\mid s^n}, ~\forall (x^n,y^n, s^n) \\
& \sum_{x_{i+1}^n}r_{x^n,y^n, (s^i, s_{i+1}^n)} = \sum_{x_{i+1}^n}r_{x^n,y^n,(s^i, s_{i+1}^n = 0^{n-i})}, \notag \\ & \qquad \forall i\in [n-1], \forall (x^i, y^n, s^n)\\
& \sum_{x_{i+1}^n}q_{x^n\mid (s^i, s_{i+1}^n)} = \sum_{x_{i+1}^n}q_{x^n\mid (s^i, s_{i+1}^n = 0^{n-i})} \notag \\ & \qquad \forall i \in [n-1], \forall (x^i, s^n)
\end{align}
Note that given any solution of LP1, the following choice of variables is a valid solution of LP2 achieving the same objective value.
\begin{align}
	&r_{x^n, y^n,s^n} = \frac{1}{M} \sum_{w\in [M]} z_{x^n,w\mid w, s^n,y^n}, \\
	&q_{x^n\mid s^n} = \frac{1}{M} \sum_{w\in [M],\widehat{w} \in [M]} z_{x^n,\widehat{w}\mid w, s^n,y^n}.
\end{align}
For the other direction, given any solution of LP2, the following choice of variables is a valid solution of LP1 achieving the same objective value.
\begin{align}
	z_{x^n,\widehat{w}\mid w,s^n,y} = 
	\begin{cases}
		r_{x^n,y^n,s^n}, & \widehat{w}=w\\
		\tfrac{1}{M-1}(q_{x^n\mid s^n} - r_{x^n,y^n,s^n}), & \widehat{w}\not= w
	\end{cases}.
\end{align} 
This shows that LP1 and LP2 have exactly the same optimal value. 
Now for the $Z_0/Z_1$ channel defined in Definition \ref{def:Z0Z1}, and $M=2,n=2$, LP2 becomes
\begin{align}
\mbox{maximize}\quad & \frac{1}{4} \sum_{x_1,x_2,y_1,y_2,s_1,s_2} r_{x_1,x_2,y_1,y_2,s_1,s_2}  \mN(y_1\mid x_1,s_1)\mN(y_2\mid x_2,s_2)\\
\mbox{s.t.}\quad
& r_{x_1,x_2,y_1,y_2,s_1,s_2}  \geq 0, ~ \forall x_1,x_2,y_1,y_2,s_1,s_2 \\
& \sum_{x_1,x_2} r_{x_1,x_2,y_1,y_2,s_1,s_2} = \frac{1}{2},~ \forall y_1,y_2,s_1,s_2  \label{eq:LP2_ex_normal_1} \\
& \sum_{x_1,x_2} q_{x_1,x_2\mid s_1,s_2} = 1, ~ \forall s_1,s_2 \label{eq:LP2_ex_normal_2}  \\
& r_{x_1,x_2,y_1,y_2,s_1,s_2} \leq q_{x_1,x_2\mid s_1,s_2}, ~ \forall x_1,x_2,y_1,y_2,s_1,s_2 \label{eq:LP2_ex_RleqQ} \\
& \sum_{x_2}r_{x_1,x_2,y_1,y_2,s_1,s_2=0} = \sum_{x_2}r_{x_1,x_2,y_1,y_2,s_1,s_2=1}, ~ \forall x_1,y_1,y_2,s_1 \label{eq:LP2_ex_causal_1} \\
& \sum_{x_2}q_{x_1,x_2\mid s_1,s_2=0} = \sum_{x_2}q_{x_1,x_2\mid s_1,s_2=1}~ \forall x_1,y_1,y_2,s_1  \label{eq:LP2_ex_causal2}
\end{align}

To obtain an upper bound for LP2, we drop the conditions in \eqref{eq:LP2_ex_causal2}, since this can only enlarge the feasible region, and thus the maximum value cannot be smaller. Let this linear program be LP3.
Then, we introduce Lagrangian multipliers $\lambda_{y_1,y_2,s_1,s_2}$ corresponding to the conditions in \eqref{eq:LP2_ex_normal_1}, $\mu_{s_1,s_2}$ corresponding to the conditions in \eqref{eq:LP2_ex_normal_2}, $\xi_{x_1,y_1,y_2,s_1}$ corresponding to the conditions in \eqref{eq:LP2_ex_causal_1}, and $\eta_{x_1,x_2,y_1,y_2,s_1,s_2}$ corresponding to the conditions in \eqref{eq:LP2_ex_RleqQ}.
The dual program thus created is
\begin{align}
	\mbox{minimize} \quad & \frac{1}{2}  \sum_{y_1,y_2,s_1,s_2} \lambda_{y_1,y_2,s_1,s_2} + \sum_{s_1,s_2} \mu_{s_1,s_2}\\
	\mbox{s.t.}\quad 
	& \lambda_{y_1,y_2,s_1,s_2} + (-1)^{s_2}\xi_{x_1,y_1,y_2,s_1} + \eta_{x_1,x_2,y_1,y_2,s_1,s_2} \geq \frac{1}{4}\mN(y_1\mid x_1,s_1)\mN(y_2\mid x_2,s_2)\notag \\
	& \hspace{2cm} \forall x_1,x_2,y_1,y_2,s_1,s_2\\
	& \mu_{s_1,s_2} \geq \sum_{y_1,y_2} \eta_{x_1,x_2,y_1,y_2,s_1,s_2}, ~\forall x_1,x_2,s_1,s_2\\
	& \eta_{x_1,x_2,y_1,y_2,s_1,s_2} \geq 0, ~ \forall x_1,x_2,y_1,y_2,s_1,s_2
\end{align}
Denote this linear program as LP4. Due to the weak duality theorem, LP4 provides an upper bound for LP3. In particular, the objective value for LP4 evaluated at any feasible point yields an upper bound on ${\eta}_{\opt,M=2,n=2}^{\NS, \ca}(\mN,\mP_{\!S})$. It can be verified that following equations specify a feasible point of LP4 (variables that are not specified in \eqref{eq:feasible_values_start}--\eqref{eq:feasible_values_end} are $0$),
\begin{align}\label{eq:feasible_values_start}
	\begin{bmatrix}
		\lambda_{0,0,0,1,0,0}\\
		\lambda_{0,0,1,0,0,0}\\
		\lambda_{0,1,0,1,0,0}\\
		\lambda_{1,0,0,1,0,0}\\
		\lambda_{1,0,1,0,0,0}\\
		\lambda_{1,1,1,0,0,0}
	\end{bmatrix}
	=
	\begin{bmatrix}
		3/16\\ 1/16 \\ 3/16 \\ 1/16 \\ 3/16 \\ 3/16
	\end{bmatrix},
\end{align}
\begin{align}
	\begin{bmatrix}
		\mu_{0,0}\\
		\mu_{0,1}\\
		\mu_{1,0}\\
		\mu_{1,1}
	\end{bmatrix}
	=
	\begin{bmatrix}
		1/8\\ 1/16\\ 1/8 \\ 1/16
	\end{bmatrix},
\end{align}
\begin{align}
	\begin{bmatrix}
		\xi_{0,0,0,0}\\
		\xi_{0,0,1,1}\\
		\xi_{0,1,0,0}\\
		\xi_{0,1,0,1}\\
		\xi_{0,1,1,1}\\
		\xi_{1,0,0,0}\\
		\xi_{1,0,0,1}\\
		\xi_{1,0,1,0}\\
		\xi_{1,1,0,0}\\
		\xi_{1,1,0,1}\\
		\xi_{1,1,1,0}\\
		\xi_{1,1,1,1}
	\end{bmatrix}
	=
	\begin{bmatrix}
		1/8 \\-1/16 \\ 1/16 \\ -1/16 \\ -1/8 \\ 1/8 \\ -1/16 \\ 1/16 \\ 1/16 \\ -1/16 \\ -1/16 \\ -3/16
	\end{bmatrix},
\end{align}
\begin{align} \label{eq:feasible_values_end}
	\begin{bmatrix}
		\eta_{0,0,0,0,0,0}\\
		\eta_{0,0,0,0,0,1}\\
		\eta_{0,0,0,0,1,0}\\
		\eta_{0,0,0,0,1,1}\\
		\eta_{0,0,0,1,1,0}\\
		\eta_{0,1,0,1,0,0}\\
		\eta_{0,1,0,1,0,1}\\
		\eta_{0,1,0,1,1,0}\\
		\eta_{0,1,0,1,1,1}\\
		\eta_{1,0,1,0,0,0}\\
		\eta_{1,0,1,0,0,1}\\
		\eta_{1,0,1,0,1,0}\\
		\eta_{1,0,1,0,1,1}\\
		\eta_{1,0,1,1,0,0}\\
		\eta_{1,1,1,1,0,0}\\
		\eta_{1,1,1,1,0,1}\\
		\eta_{1,1,1,1,1,0}\\
		\eta_{1,1,1,1,1,1}
	\end{bmatrix}
	=
	\begin{bmatrix}
		1/8\\1/16\\1/16\\1/16\\1/16\\1/8\\1/16\\1/8\\1/16\\1/16\\1/16\\1/8\\1/16\\1/16\\1/8\\1/16\\1/8\\1/16
	\end{bmatrix},
\end{align}
and that the objective value of LP4 attained at this point is $13/16$.
We thus conclude the proof. \hfill \qed

\section{Example}\label{sec:toyexample}
In this section, we provide an example of a NS-assisted coding scheme with causal CSIT.  Let us consider the channel $\mN$ to be defined by $\mN(y\mid x,s) = \mathbb{I}[y=xs]$, for $x\in \{0,1\}, y\in \{0,1\}, s\in \{0,1\}$. 
Equivalently, the input-output relationship of this channel can be written as $Y=XS$, i.e., $Y$ equals the product of $X$ and $S$. Suppose the channel $\mN$ is used $3$ times. 

The aim of this example is to show the construction of a NS-assisted coding scheme with causal CSIT. For this toy example, let us deviate from our standard problem formulation in order to avoid an asymptotic blocklength and typicality arguments. To this end, let us assume (only in this example) that the channel states $(S_1,S_2,S_3)$ are distributed uniformly in $\mathfrak{S} \triangleq \{(0,1,1), (1,0,1), (1,1,0)\}$, and therefore they are no longer independent. 
In particular, if $S_1=0$, then the transmitter learns $(S_1,S_2,S_3)$ immediately at the first time slot, as $(S_2,S_3)$ can only be $(1,1)$. If $S_1=1$, then at the first time slot $(S_2,S_3)$ is still equally likely to be either $(0,1)$ or $(1,0)$. In whichever case, the transmitter will learn $(S_1,S_2,S_3)$ at the second time slot, after $(S_1,S_2)$ are revealed. 

Let ${\bf x} = (x_1,x_2,x_3), {\bf y} =(y_1,y_2,y_3)$ and ${\bf s} = (s_1,s_2,s_3)$. Recall that a NS-assisted coding scheme with causal CSIT and message size $M$ is defined by a conditional distribution $\mZ({\bf x},\widehat{w}\mid w, {\bf s}, {\bf y})$ for $w \in [M], \widehat{w} \in [M], {\bf x} \in \{0,1\}^3, {\bf y} \in \{0,1\}^3$ and ${\bf s} \in \{0,1\}^3$. 
Let us first assume that, for fixed ${\bf x}, \widehat{w}, w$ and ${\bf y}$,
\begin{enumerate}[label=\textit{A\arabic*:}]
	\item $\mZ({\bf x},\widehat{w}\mid w, (0,s_2,s_3), {\bf y}) = \mZ({\bf x},\widehat{w}\mid w, (0,1,1), {\bf y})$ for all $(s_2, s_3) \in \{0,1\}^2\setminus \{(1,1)\}$;
	\item $\mZ({\bf x},\widehat{w}\mid w, (1,0,0), {\bf y}) = \mZ({\bf x},\widehat{w}\mid w, (1,0,1), {\bf y})$;
	\item $\mZ({\bf x},\widehat{w}\mid w, (1,1,1), {\bf y}) = \mZ({\bf x},\widehat{w}\mid w, (1,1,0), {\bf y})$.
\end{enumerate}
We then only need to specify $\mZ({\bf x},\widehat{w}\mid w, {\bf s}, {\bf y})$ for $w \in [M], \widehat{w} \in [M], {\bf x} \in \{0,1\}^3, {\bf y} \in \{0,1\}^3$ and ${\bf s} \in \mathfrak{S}$, and the remaining values, i.e., $\mZ({\bf x},\widehat{w}\mid w, {\bf s}, {\bf y})$ for ${\bf s} \in \{0,1\}^3\setminus \mathfrak{S}$, are defined accordingly based on  \textit{A1}--\textit{A3}. 
Specifically, \textit{A1} corresponds to the case when $S_1=0$. \textit{A2} corresponds to the case when $S_1=1$ and $S_2=0$. \textit{A3} corresponds to the case when $S_1=1$ and $S_2=1$. 
Besides being a valid conditional distribution, the scheme $\mZ$ needs to also satisfy the NS conditions \textit{C1} and \textit{C2}, and the condition of causal CSIT \textit{C3}. 

We will construct a scheme with $M=4$, i.e., sending a message of 2 bits. First, according to the chain rule of conditional probability, $\mZ$ can be factorized as
\begin{align} \label{eq:ex_factor}
	\mZ({\bf x},\widehat{w}\mid w, {\bf s}, {\bf y}) = \mZ({\bf x}\mid w, {\bf s}, {\bf y}) \times \mZ(\widehat{w}\mid {\bf x}, w,  {\bf s}, {\bf y}).
\end{align}
We set the first factor
\begin{align} \label{eq:ex_first}
	\mZ({\bf x}\mid w, {\bf s}, {\bf y}) = 1/8
\end{align}
for ${\bf x}\in \{0,1\}^3, w\in [4], {\bf y}\in \{0,1\}^3$ and ${\bf s}\in \mathfrak{S}$. The meaning of \eqref{eq:ex_first} is to let the scheme generate $X_i$ i.i.d. uniform, regardless of the values of $w,{\bf s}$ and ${\bf y}$. This immediately guarantees \textit{C1}. 

We then set the second factor
\begin{align} \label{eq:ex_second}
	\mZ(\widehat{w}\mid {\bf x}, w, {\bf s}, {\bf y}) = \begin{cases} T({\bf x}, {\bf y}, {\bf s}), & \widehat{w}= w \\
		\tfrac{1}{3} \big( 1- T({\bf x}, {\bf y}, {\bf s}) \big), & \widehat{w} \neq w
	\end{cases}
\end{align}
for $\widehat{w}\in [4], {\bf x}\in \{0,1\}^3, w\in [4], {\bf y} \in \{0,1\}^3$ and ${\bf s} \in \mathfrak{S}$, where
\begin{align} \label{eq:ex_def_T}
	T({\bf x}, {\bf y}, {\bf s}) \triangleq \prod_{i\colon s_i =1}\mathbb{I}\big[ y_i = x_is_i \big] 
\end{align}
for ${\bf x}\in \{0,1\}^3, {\bf y} \in \{0,1\}^3$ and ${\bf s} \in \mathfrak{S}$.
The meaning of \eqref{eq:ex_second} is explained as follows. Firstly, consider $T$ as an authentication process on the input $({\bf x}, {\bf y}, {\bf s})$, which returns $1$ if $y_i = x_i s_i$ at the two positions where $s_i=1$, and returns $0$ otherwise. Then, if $T({\bf x}, {\bf y}, {\bf s}) = 1$, the scheme outputs $\widehat{W} = w$, i.e., the correct message. If $T({\bf x}, {\bf y}, {\bf s}) = 0$, the scheme outputs $\widehat{W}\sim {\rm Unif}([M]\setminus \{w\})$, i.e., a uniformly distributed  incorrect message. 
One can check from \eqref{eq:ex_second} that $\mZ(\widehat{w}\mid {\bf x}, w,  {\bf s}, {\bf y})$ is a valid conditional distribution, since for each $w\in [4]$ there are $3$ elements from $[4]$ that is not equal to $w$ and one element being equal to $w$, so $\sum_{\widehat{w}\in [4]} \mZ(\widehat{w}\mid {\bf x}, w, {\bf y}, {\bf s}) = 1$.

We now prove that the scheme $\mZ$ thus defined also satisfies \textit{C2} and \textit{C3}. Due to \textit{A1}--\textit{A3}, for \textit{C2}, it suffices to check for $\widehat{w}\in [4], w\in [4], {\bf s} \in \mathfrak{S}$ and ${\bf y}\in \{0,1\}^3$, $\sum_{{\bf x}\in \{0,1\}^3} \mZ({\bf x},\widehat{w}\mid w, {\bf s}, {\bf y})$ is invariant under changes of $w\in [4]$ and ${\bf s}\in \mathfrak{S}$. This is shown as follows.
\begin{align}
	&\sum_{{\bf x}\in \{0,1\}^3} \mZ({\bf x},\widehat{w}\mid w, {\bf s}, {\bf y}) \notag \\
	&\stackrel{\eqref{eq:ex_factor},\eqref{eq:ex_first}}{=} \frac{1}{8} \sum_{{\bf x}\in \{0,1\}^3} \mZ(\widehat{w}\mid w, {\bf x}, {\bf s}, {\bf y}) \\
	&\stackrel{\eqref{eq:ex_second}}{=} \begin{cases}
		\frac{1}{8} \sum_{{\bf x}\in \{0,1\}^3} T({\bf x}, {\bf y}, {\bf s}), & \widehat{w} = w \\
		\frac{1}{8} \sum_{{\bf x}\in \{0,1\}^3} \tfrac{1}{3}\big( 1- T({\bf x}, {\bf y}, {\bf s}) \big), & \widehat{w} \neq w
	\end{cases} \\
	& = \frac{1}{4}
\end{align}
To see the last step, note that for ${\bf s} \in \mathfrak{S}$, there are exactly two positions where $s_i = 1$, and one position where $s_i=0$. Due to symmetry it suffices to consider $(s_1,s_2,s_3) = (0,1,1)$. Then $T({\bf x}, {\bf y}, {\bf s}) = 1$ if and only if $(x_2,x_3) = (y_2,y_3)$, i.e., ${\bf x} = (0,y_2,y_3)$ or ${\bf x} = (1,y_2,y_3)$. Therefore, we have $\frac{1}{8} \sum_{{\bf x}\in \{0,1\}^3} T({\bf x}, {\bf y}, {\bf s}) = \frac{1}{4}$, and $\frac{1}{8} \sum_{{\bf x}\in \{0,1\}^3} \tfrac{1}{3}\big( 1- T({\bf x}, {\bf y}, {\bf s}) \big) = \frac{1}{8}\cdot \frac{1}{3}(8-\sum_{{\bf x}\in \{0,1\}^3}T({\bf x}, {\bf y}, {\bf s})) = \frac{1}{4}$. 

For \textit{C3}, again due to \textit{A1}--\textit{A3}, it suffices to check 
\begin{equation}
\begin{aligned}
	\sum_{(x_2,x_3)\in \{0,1\}^2} \mZ((x_1,x_2,x_3),\widehat{w}\mid w, (1,0,1), {\bf y}) \\ = \sum_{(x_2,x_3)\in \{0,1\}^2} \mZ((x_1,x_2,x_3),\widehat{w}\mid w, (1,1,0), {\bf y})
\end{aligned}
\end{equation}
for all $x_1\in \{0,1\}, \widehat{w} \in [M], w \in [M]$ and ${\bf y} \in \{0,1\}^3$. Due to \eqref{eq:ex_factor}, \eqref{eq:ex_first} and \eqref{eq:ex_second}, it suffices to show
\begin{equation}\label{eq:ex_check_C3}
\begin{aligned}
	\sum_{(x_2,x_3)\in \{0,1\}^2} T({\bf x}, {\bf y}, (1,0,1))  \\
	=\sum_{(x_2,x_3)\in \{0,1\}^2} T({\bf x}, {\bf y}, (1,1,0))
\end{aligned}
\end{equation}
for all $x_1\in \{0,1\}$ and ${\bf y} \in \{0,1\}^3$.
Note that if $x_1\neq y_1$, then both the LHS and the RHS of \eqref{eq:ex_check_C3} are equal to $0$. If $x_1=y_1$, then both the LHS and the RHS of \eqref{eq:ex_check_C3} are equal to $2$. Therefore, \textit{C3} is satisfied.
The scheme $\mZ$ thus satisfies both the NS and the causality constraints. 

Finally, let us show how the scheme works when it is connected with the channel. Conditioned on $(S_1,S_2,S_3) = (0,1,1)$, the authentication will always pass since the channel guarantees that $Y_2 =  X_2S_2$ and $Y_3= X_3S_3$. In this case the output message $\widehat{W} = W$. The cases for $(S_1,S_2,S_3) = (1,0,1)$ and $(S_1,S_2,S_3) = (1,1,0)$ are similarly argued. Therefore, the scheme will always output $\widehat{W} = W$ at the receiver when the channel is connected.

\section{Proof of Theorem \ref{thm:strictly_causal} \label{proof:strictly_causal}}
Since any rate $R< \max_{\mP_{\! X\mid S}} I(X;Y\mid S)$ is achievable even without the transmitter knowing $T$, the achievability proof for Theorem \ref{thm:strictly_causal} follows from previous results (e.g., Theorem \ref{thm:NS_capacity}). Therefore, it suffices to provide the converse proof for Theorem \ref{thm:strictly_causal}, i.e., when $S$ is known non-causally, and $T$ is known strictly causally to the transmitter, the capacity is still not more than $\max_{\mP_{\! X\mid S}} I(X;Y\mid S)$.
The key of the proof is to construct a channel-absent distribution (under which the decoded message $\widehat{W}$ is independent of the message $W$) and applies the data-processing inequality for relative entropy.

Let $\mZ \in \mathcal{Z}^{\NS, \mixed}(M,n)$ be any NS-assisted coding scheme described in Section \ref{sec:NSschemes_mixed}. Recall from \eqref{eq:joint_prob_strc} that the joint distribution of $(W,S^n,T^n, X^n,Y^n, \widehat{W})$ is,
\begin{align}
	&\sfp_{W S^nT^nX^n Y^n \widehat{W}}(w, s^n, t^n, x^n, y^n,  \widehat{w}) \notag \\
	&= \frac{1}{M} \Bigg(\prod_{i=1}^n  \mP_{\! S}(s_i) \mP_{\! T}(t_i) \mN(y_i\mid x_i,s_i,t_i)\Bigg)   \mZ \big(x^n, \widehat{w} \mid w,s^{n}, t^{n-1}, y^n\big)  \notag 
\end{align}
We also know that $T_i$ is independent of $(X^i,S^n), \forall i\in [n]$. This follows from the definition that $S^n$ and $T^n$ are independent i.i.d. sequences, and that for each $i\in [n]$, $X^i$ is produced by the scheme up to the $i^{th}$ channel use, at which point $T_i$ has not been provided to the scheme. We therefore have
\begin{align}
	&\mP_{\! T}(t_i) \mN(y_i\mid x_i,s_i,t_i) \notag \\
	&= \sfp(t_i) \sfp(y_i\mid x_i,s_i,t_i) \\
	&= \sfp(t_i\mid x_i,s_i) \sfp(y_i\mid x_i,s_i,t_i) \\
	&= \sfp(t_i, y_i \mid x_i,s_i) \\
	&= \sfp(y_i\mid x_i,s_i)  \sfp(t_i\mid x_i, y_i,s_i)
\end{align}
and rewrite
\begin{align}
	&\sfp_{W S^nT^n X^n Y^n \widehat{W}}(w, s^n, t^n, x^n, y^n,  \widehat{w}) \notag \\
	&=\frac{1}{M} \Bigg(\prod_{i=1}^n  \mP_{\! S}(s_i) \sfp(y_i\mid x_i,s_i) \sfp(t_i\mid x_i, y_i,s_i) \Bigg)   \mZ \big(x^n, \widehat{w} \mid w,s^{n}, t^{n-1}, y^n\big)
\end{align}

To apply the data-processing inequality, consider an alternative distribution  $\sfq_{WS^nT^n X^nY^n\widehat{W}} \in \mathcal{P}([M] \times \mathcal{S}^n \times \mathcal{T}^n \times \mathcal{X}^n  \times \mathcal{Y}^n \times [M])$, defined as
\begin{align}
	&\sfq_{WS^nT^n X^nY^n\widehat{W}}(w,s^n,t^n, x^n,y^n,\widehat{w}) \notag \\
	&\triangleq \frac{1}{M} \Bigg(\prod_{i=1}^n \mP_{\! S}(s_i) \sfp(y_i\mid s_i) \sfp(t_i\mid x_i, y_i,s_i) \Bigg) \mZ \big(x^n, \widehat{w} \mid w,s^{n}, t^{n-1}, y^n\big)
\end{align}
Note that the only difference in $\sfq$ from $\sfp$ is that $\sfp(y_i\mid x_i,s_i)$ is replaced by $\sfp(y_i\mid s_i)$. We now argue that $W$ and $\widehat{W}$ are independent under this alternative distribution $\sfq$, i.e, 
\begin{align} \label{eq:indep_WW_hat_q}
	\sfq_{W\widehat{W}}(w,\widehat{w}) = \sfq_{W}(w)\sfq_{\widehat{W}}(\widehat{w}) = \frac{1}{M}\sfq_{\widehat{W}}(\widehat{w}) = \sfp_{W}(w)\sfq_{\widehat{W}}(\widehat{w}).
\end{align}
Intuitively, $\sfq$ corresponds to a setting where $(W,S^n,T^n,X^n,Y^n,\widehat{W})$ are generated in the following manner. Suppose $(S^n,Y^n)$ are generated first at a third party according to the distribution $\prod_{i=1}^n\mP_S(s_i)\sfp(y_i\mid s_i)$.  $(S^n,Y^n)$ is then revealed to both the transmitter and the receiver. This process happens before the communication.
After this, the transmitter generates $W$, and obtains $X_1$ from the scheme by providing it with $(W,S^n)$. It then generates $T_1$ according to the marginal distribution $\sfp(t_1\mid x_1,y_1,s_1)$. Sequentially for $i=2,\ldots, n$, the transmitter obtains $X_i$ from the scheme by providing it with $T_{i-1}$, and then generates $T_i$ according to the marginal distribution $\sfp(t_i\mid x_i,y_i,s_i)$. Meanwhile, the receiver obtains $\widehat{W}$ from the scheme by providing it with $Y^n$. 
During the whole process the channel with state is absent, and $(S^n,Y^n)$ is not more than a form of shared randomness. In fact the receiver may obtain $\widehat{W}$ even before $W$ is generated. Therefore, $W$ and $\widehat{W}$ must be independent under $\sfq$.  We delegate the formal proof of $\sfq(w,\widehat{w}) = \frac{1}{M}\sfq(\widehat{w})$ to the end of this section.

Proceeding along the lines of \cite{Polyanskiy_Poor_Verdu} let us consider the relative entropy,
\begin{align}
	&D\big(\sfp_{WS^nT^nX^nY^n\widehat{W}} \Vert \sfq_{WS^nT^nX^nY^n\widehat{W}}\big) \notag \\
	&=\mathbb{E}_{\sfp}\Bigg[ \log_2 \Bigg( \frac{\sfp(W,S^n,T^n,X^n,Y^n,\widehat{W})}{\sfq(W,S^n,T^n,X^n,Y^n,\widehat{W})} \Bigg) \Bigg] \\
	&=\mathbb{E}_{\sfp}\Bigg[ \log_2 \Bigg( \prod_{i=1}^n \frac{\sfp(Y_i\mid X_i,S_i)}{\sfp(Y_i\mid S_i)} \Bigg) \Bigg] \\
	&= \sum_{i=1}^n I_{\sfp}(X_i;Y_i\mid S_i) \label{eq:conv_1}
\end{align}

On the other hand, 
\begin{align}
	&D\big(\sfp_{WS^nT^nX^nY^n\widehat{W}} \Vert \sfq_{WS^nT^nX^nY^n\widehat{W}}\big) \notag \\
	&\geq D(\sfp_{W\widehat{W}} \Vert \sfq_{W\widehat{W}}) \label{eq:use_dp} \\
	&=D(\sfp_{W\widehat{W}} \Vert \sfp_{W}\sfq_{\widehat{W}}) \label{eq:use_indp} \\
	&\geq D(\sfp_{W\widehat{W}} \Vert \sfp_{W}\sfp_{\widehat{W}}) \label{eq:use_nn_D} \\
	&= I_{\sfp}(W;\widehat{W}) \label{eq:conv_2}
\end{align}
Step \eqref{eq:use_dp} applies the data-processing inequality for relative entropy (e.g., \cite[Lem. 3.11]{csiszar2011information}) to a channel which maps $WS^nT^nX^nY^n\widehat{W} \to W\widehat{W}$.
Step \eqref{eq:use_indp} uses \eqref{eq:indep_WW_hat_q}. Step \eqref{eq:use_nn_D} is because $D(\sfp_{W\widehat{W}} \Vert \sfp_{W}\sfq_{\widehat{W}}) = \mathbb{E}_{\sfp}[\log_2 \big( \frac{\sfp(W,\widehat{W})}{\sfp(W)\sfq(\widehat{W})} \big)] = \mathbb{E}_{\sfp}[\log_2 \big( \frac{\sfp(W,\widehat{W})}{\sfp(W)\sfp(\widehat{W})} \big)]+\mathbb{E}_{\sfp}[\log_2 \big( \frac{\sfp(\widehat{W})}{\sfq(\widehat{W})} \big)] = D(\sfp_{W\widehat{W}} \Vert \sfp_{W}\sfp_{\widehat{W}}) + D(\sfp_{\widehat{W}} \Vert \sfq_{\widehat{W}})$ together with the non-negativity of the KL divergence $D(\sfp_{\widehat{W}} \Vert \sfq_{\widehat{W}}) \geq 0$.
Combining \eqref{eq:conv_1} and \eqref{eq:conv_2}, we have that 
\begin{align} \label{eq:singleletter}
	I_{\sfp}(W;\widehat{W}) \leq \sum_{i=1}^n I_{\sfp}(X_i;Y_i\mid S_i) \leq n \max_{\mP_{\!X\mid S}}I(X;Y\mid S)
\end{align}
where the maximization is over all $\mP_{\!X\mid S} \in \mathcal{P}(\mathcal{X}\mid \mathcal{S})$  such that $(S,X,Y)\sim \mP_{\!S}(s) \mP_{\!X}(x\mid s) \mN(y\mid x,s)$.

For any coding schemes $\mZ_n \in \mathcal{Z}^{\NS,\mixed}(M_n,n)$ for which the probability of success $\lim_{n\to \infty} \eta(\mZ_n) = 1$ and $\lim_{n\to \infty} \frac{\log_2(M_n)}{n} \geq R$, Fano's inequality implies that
\begin{align} \label{eq:use_Fano}
	R \leq \lim_{n\to \infty} \frac{1}{n} I_{\sfp}(W;\widehat{W}) \stackrel{\eqref{eq:singleletter}}{\leq} \max_{\mP_{\!X\mid S}}I(X;Y\mid S).
\end{align}
Thus, we conclude that  $C^{\NS,\mixed} \leq \max_{\mP_{\!X\mid S}}I(X;Y\mid S)$. \hfil \qed

Let us prove that $\sfq(w,\widehat{w}) = \frac{1}{M}\sfq(\widehat{w})$. This is done by recursively using the properties of the scheme $\mZ$.
\begin{align}
	&\sfq(w,\widehat{w}) \notag \\
	&=\sum_{y^n,s^n,x^n,t^n}\frac{1}{M} \Bigg(\prod_{i=1}^n \mP_{\! S}(s_i) \sfp(y_i\mid s_i) \sfp(t_i\mid x_i, y_i,s_i) \Bigg) \mZ \big(x^n, \widehat{w} \mid w,s^{n}, t^{n-1}, y^n\big) \\
	&=\frac{1}{M}\sum_{y^n,s^n}\sfq(s^n,y^n) \sum_{x^n,t^n}   \Bigg(\prod_{i=1}^{n} \sfp(t_i\mid x_i, y_i,s_i)\Bigg)  \mZ \big(x^n, \widehat{w} \mid w,s^{n}, t^{n-1}, y^n\big)\\
	&=\frac{1}{M}\sum_{y^n,s^n}\sfq(s^n,y^n) \sum_{x^n,t^{n-1}}   \Bigg(\prod_{i=1}^{n-1} \sfp(t_i\mid x_i, y_i,s_i)\Bigg)  \mZ \big(x^n, \widehat{w} \mid w,s^{n}, t^{n-1}, y^n\big)\\
	&\stackrel{\eqref{eq:cond_strc_3}}{=}\frac{1}{M}\sum_{y^n,s^n}\sfq(s^n,y^n) \sum_{x^{n-1},t^{n-1}}   \Bigg(\prod_{i=1}^{n-1} \sfp(t_i\mid x_i, y_i,s_i)\Bigg)  \mZ \big(x^{n-1}, \widehat{w} \mid w,s^{n}, t^{n-2}, y^n\big)\\
	&~~~~\vdots \\
	&\stackrel{\eqref{eq:cond_strc_3}}{=}\frac{1}{M}\sum_{y^n,s^n}\sfq(s^n,y^n) \sum_{x_1,t_1}   \Bigg( \sfp(t_1\mid x_1, y_1,s_1)\Bigg)  \mZ \big(x_1, \widehat{w} \mid w,s^n, y^n\big)\\
	&=\frac{1}{M}\sum_{y^n,s^n}\sfq(s^n,y^n) \sum_{x_1}     \mZ \big(x_1, \widehat{w} \mid w,s^n, y^n\big)\\
	&\stackrel{\eqref{eq:cond_strc_2}}{=}\frac{1}{M}\sum_{y^n,s^n}\sfq(s^n,y^n)    \mZ \big(\widehat{w} \mid y^n \big)\\
	&=\frac{1}{M}\sum_{y^n}\sfq(y^n)    \sfq \big(\widehat{w} \mid y^n \big)\\
	&=\frac{1}{M}\sfq (\widehat{w})
\end{align}
\hfil \qed

\bibliographystyle{IEEEtran}
\bibliography{../bib_file/yy.bib}

\begin{thebibliography}{10}
\providecommand{\url}[1]{#1}
\csname url@samestyle\endcsname
\providecommand{\newblock}{\relax}
\providecommand{\bibinfo}[2]{#2}
\providecommand{\BIBentrySTDinterwordspacing}{\spaceskip=0pt\relax}
\providecommand{\BIBentryALTinterwordstretchfactor}{4}
\providecommand{\BIBentryALTinterwordspacing}{\spaceskip=\fontdimen2\font plus
\BIBentryALTinterwordstretchfactor\fontdimen3\font minus \fontdimen4\font\relax}
\providecommand{\BIBforeignlanguage}[2]{{%
\expandafter\ifx\csname l@#1\endcsname\relax
\typeout{** WARNING: IEEEtran.bst: No hyphenation pattern has been}%
\typeout{** loaded for the language `#1'. Using the pattern for}%
\typeout{** the default language instead.}%
\else
\language=\csname l@#1\endcsname
\fi
#2}}
\providecommand{\BIBdecl}{\relax}
\BIBdecl

\bibitem{lapidothCommonState}
A.~Lapidoth and Y.~Steinberg, ``The multiple-access channel with causal side information: Common state,'' \emph{IEEE Transactions on Information Theory}, vol.~59, no.~1, pp. 32--50, 2012.

\bibitem{shannon_CSIT}
C.~E. Shannon, ``Channels with side information at the transmitter,'' \emph{IBM Journal of Research and Development}, vol.~2, no.~4, pp. 289--293, 1958.

\bibitem{gel1980coding}
S.~Gel'fand and M.~Pinsker, ``Coding for channels with random parameters,'' \emph{Probl. Contr. Inform. Theory}, vol.~9, no.~1, pp. 19--31, 1980.

\bibitem{leditzky2020playing}
F.~Leditzky, M.~A. Alhejji, J.~Levin, and G.~Smith, ``Playing games with multiple access channels,'' \emph{Nature communications}, vol.~11, no.~1, p. 1497, 2020.

\bibitem{seshadri2023separation}
A.~Seshadri, F.~Leditzky, V.~Siddhu, and G.~Smith, ``On the separation of correlation-assisted sum capacities of multiple access channels,'' \emph{IEEE Transactions on Information Theory}, vol.~69, no.~9, pp. 5805--5844, 2023.

\bibitem{Bennett_Shor_Smolin_Thapliyal_PRL}
\BIBentryALTinterwordspacing
C.~H. Bennett, P.~W. Shor, J.~A. Smolin, and A.~V. Thapliyal, ``Entanglement-assisted classical capacity of noisy quantum channels,'' \emph{Phys. Rev. Lett.}, vol.~83, pp. 3081--3084, Oct 1999. [Online]. Available: \url{https://link.aps.org/doi/10.1103/PhysRevLett.83.3081}
\BIBentrySTDinterwordspacing

\bibitem{cubitt2011zero}
T.~S. Cubitt, D.~Leung, W.~Matthews, and A.~Winter, ``Zero-error channel capacity and simulation assisted by non-local correlations,'' \emph{IEEE Transactions on Information Theory}, vol.~57, no.~8, pp. 5509--5523, 2011.

\bibitem{agarwal2024nonlocalityassistedenhancementerrorfreecommunication}
\BIBentryALTinterwordspacing
K.~Agarwal, S.~G. Naik, A.~Chakraborty, S.~Sen, P.~Ghosal, B.~Paul, M.~Banik, and R.~K. Patra, ``Nonlocality-assisted enhancement of error-free communication in noisy classical channels,'' 2024. [Online]. Available: \url{https://arxiv.org/abs/2412.04779}
\BIBentrySTDinterwordspacing

\bibitem{matthews2012linear}
W.~Matthews, ``A linear program for the finite block length converse of {Polyanskiy–Poor–Verdú} via nonsignaling codes,'' \emph{IEEE Transactions on Information Theory}, vol.~58, no.~12, pp. 7036--7044, 2012.

\bibitem{fawzi2024MAC}
O.~Fawzi and P.~Fermé, ``Multiple-access channel coding with non-signaling correlations,'' \emph{IEEE Transactions on Information Theory}, vol.~70, no.~3, pp. 1693--1719, 2024.

\bibitem{fawzi2024broadcast}
------, ``Broadcast channel coding: Algorithmic aspects and non-signaling assistance,'' \emph{IEEE Transactions on Information Theory}, vol.~70, no.~11, pp. 7563--7580, 2024.

\bibitem{Yao_Jafar_NS_DoF}
Y.~Yao and S.~A. Jafar, ``Can non-signaling assistance increase the degrees of freedom of a wireless network?'' \emph{IEEE Transactions on Information Theory}, vol.~72, no.~2, pp. 844--864, 2026.

\bibitem{Quek_Shor}
\BIBentryALTinterwordspacing
Y.~Quek and P.~W. Shor, ``Quantum and superquantum enhancements to two-sender, two-receiver channels,'' \emph{Phys. Rev. A}, vol.~95, p. 052329, May 2017. [Online]. Available: \url{https://link.aps.org/doi/10.1103/PhysRevA.95.052329}
\BIBentrySTDinterwordspacing

\bibitem{Yao_Jafar_CSITTP}
\BIBentryALTinterwordspacing
Y.~Yao and S.~A. Jafar, ``Virtual signaling of {CSIT} via non-signaling assistance,'' \emph{ArXiv:2506.17803}, 2025. [Online]. Available: \url{https://arxiv.org/abs/2506.17803}
\BIBentrySTDinterwordspacing

\bibitem{NIT}
A.~El~Gamal and Y.-H. Kim, \emph{Network information theory}.\hskip 1em plus 0.5em minus 0.4em\relax Cambridge University Press, 2011.

\bibitem{gallego2014nonlocality}
R.~Gallego, L.~E. W{\"u}rflinger, R.~Chaves, A.~Ac{\'\i}n, and M.~Navascu{\'e}s, ``Nonlocality in sequential correlation scenarios,'' \emph{New Journal of Physics}, vol.~16, no.~3, p. 033037, 2014.

\bibitem{TONS_boxes}
R.~Ramanathan, M.~Banacki, R.~Ravell~Rodr{\'\i}guez, and P.~Horodecki, ``Single trusted qubit is necessary and sufficient for quantum realization of extremal no-signaling correlations,'' \emph{npj Quantum Information}, vol.~8, no.~1, p. 119, 2022.

\bibitem{Polyanskiy_Poor_Verdu}
Y.~Polyanskiy, H.~V. Poor, and S.~Verdu, ``Channel coding rate in the finite blocklength regime,'' \emph{IEEE Transactions on Information Theory}, vol.~56, no.~5, pp. 2307--2359, May 2010.

\bibitem{csiszar2011information}
I.~Csisz{\'a}r and J.~K{\"o}rner, \emph{Information theory: coding theorems for discrete memoryless systems}.\hskip 1em plus 0.5em minus 0.4em\relax Cambridge University Press, 2011.

\end{thebibliography}
\end{document}